\documentclass[10pt]{IEEEtran}
\usepackage[T1]{fontenc}
\usepackage[latin9]{inputenc}
\usepackage{calc}
\usepackage{amsmath}
\usepackage{amssymb}
\usepackage{graphicx}

\makeatletter

\providecommand{\tabularnewline}{\\}

  \newtheorem{definitn}{Definition}
  \newtheorem{problem}{Problem}
  \newtheorem{remrk}{Remark}
  \newtheorem{thm}{Theorem}
  \newtheorem{cor}{Corollary}
  \newtheorem{lemma}{Lemma}

\usepackage{cite}

\author{Xiongbin~Rao,~\IEEEmembership{Student~Member,~IEEE} and~Vincent~K.~N.~Lau,~\IEEEmembership{Fellow,~IEEE}
\thanks{The authors are with the Department of Electronic and Computer Engineering (ECE), the Hong Kong University of Science and Technology (HKUST), Hong Kong (e-mail: \{xrao,eeknlau\}@ust.hk).}%
 }

\makeatother

\begin{document}

\title{Distributed Compressive CSIT Estimation and Feedback for FDD Multi-user
Massive MIMO Systems}
\maketitle
\begin{abstract}
To fully utilize the spatial multiplexing gains or array gains of
massive MIMO, the channel state information must be obtained at the
transmitter side (CSIT). However, conventional CSIT estimation approaches
are not suitable for FDD massive MIMO systems because of the overwhelming
training and feedback overhead. In this paper, we consider multi-user
massive MIMO systems and deploy the compressive sensing (CS) technique
to reduce the training as well as the feedback overhead in the CSIT
estimation. The multi-user massive MIMO systems exhibits a hidden
joint sparsity structure in the user channel matrices due to the shared
local scatterers in the physical propagation environment. As such,
instead of naively applying the conventional CS to the CSIT estimation,
we propose a distributed compressive CSIT estimation scheme so that
the compressed measurements are observed at the users locally, while
the CSIT recovery is performed at the base station jointly. A joint
orthogonal matching pursuit recovery algorithm is proposed to perform
the CSIT recovery, with the capability of exploiting the hidden joint
sparsity in the user channel matrices. We analyze the obtained CSIT
quality in terms of the normalized mean absolute error, and through
the closed-form expressions, we obtain simple insights into how the
joint channel sparsity can be exploited to improve the CSIT recovery
performance. \end{abstract}

\section{Introduction }

Massive multiple-input multiple-output (MIMO) can greatly enhance
the wireless communication capacity due to the increased degrees of
freedom \cite{telatar1999capacity}, and there is intense research
interest in the applications of massive MIMO in next generation wireless
systems \cite{larsson2013massive}. To fully utilize the spatial multiplexing
gains and the array gains of massive MIMO \cite{shi2011iteratively,bogale2012weighted},
knowledge of channel state information at the transmitter (CSIT) is
essential. In time-division duplexing (TDD) massive MIMO systems,
the CSIT can be obtained by exploiting the channel reciprocity using
uplink pilots \cite{hoydis2013massive} and hence, lots of works today
\cite{larsson2013massive,nguyen2013compressive} have considered massive
MIMO of TDD systems. On the other hand, as frequency-division duplexing
(FDD) is generally considered to be more effective for systems with
symmetric traffic and delay-sensitive applications \cite{chan2006evolution}
and the most cellular systems today employ FDD, it is therefore of
great interest to explore effective approaches for obtaining CSIT
for massive MIMO with FDD \cite{marzetta2013special}. To obtain CSIT
at the base station (BS) of FDD systems, the BS first transmits downlink
pilot symbols so that the user can estimate the downlink CSI locally.
The estimated CSI are then fed back to the BS via uplink signaling
channels \cite{sun2002estimation}. Conventional methods to estimate
the downlink CSI at the users include least square (LS) \cite{biguesh2006training}
and minimum mean square error (MMSE) \cite{yin2013coordinated}. However,
using these conventional CSI estimation techniques, the number of
independent pilot symbols required at the BS has to scale linearly
with the number of transmit antennas $M$ at the BS (i.e. $O(M)$).
For massive MIMO, as $M$ becomes very large, the pilot training overhead
(downlink) as well as the CSI feedback overhead (uplink) would be
prohibitively large. In addition, the number of independent pilot
symbols available is limited by the channel coherence time and coherence
bandwidth \cite{larsson2013massive}, as illustrated in Figure \ref{fig:Frame-structure-with}.
Obviously, as $M$ increases in FDD massive MIMO systems, we do not
have sufficient pilots to support CSIT estimation despite the overhead
issues. Hence, a new CSIT estimation and feedback design will be needed
to support FDD massive MIMO systems. 

There is one important observation for massive MIMO systems, which
can help to address the above issues. From many experimental studies
of massive MIMO channels \cite{zhou2006experimental,kyritsi2003correlation,kaltenberger2008correlation,hoydis2012channel,gao2011linear},
as $M$ increases, the user channel matrices tend to be sparse due
to the limited local scatterers at the BS. Hence, it is very inefficient
to estimate the entire CSI matrices using long pilot training symbols
at the BS. Instead, we should exploit the hidden sparsity in the CSIT
estimation and feedback process and compressive sensing (CS) is an
attractive framework for this purpose \cite{berger2010application}.
In fact, the CS techniques have already been used in the literature
to enhance the channel estimation performance. For instance, in \cite{nguyen2013compressive},
a CS-based low-rank approximation algorithm is proposed to enhance
the channel estimation performance for TDD massive MIMO systems but
the technique cannot be applied for FDD systems. In \cite{bajwa2010compressed},
a CS-based channel estimation method is proposed to exploit the per-link
sparse multipath channels in time, frequency as well as spatial domains
in MIMO systems. By exploiting the spatial sparsity using CS in massive
MIMO systems, it is shown that only $O(s\log M)$ training%
\footnote{$s$ denotes the sparsity level, i.e., the number of non-zero spatial
channel paths.%
} overhead \cite{bajwa2010compressed} is needed and this represents
a substantial reduction of the CSIT estimation overhead compared with
the conventional LS approach. To extend existing CS-based CSIT estimation
techniques to multi-user massive MIMO of FDD systems, we need to address
several first order technical challenges:
\begin{itemize}
\item \textbf{How to exploit the joint channel sparsity among different
users distributively. }As has been observed in many experimental studies
\cite{kaltenberger2008correlation,kyritsi2003correlation,hoydis2012channel,gao2011linear},
the user channel matrices of a multi-user massive MIMO system may
be jointly correlated due to the shared common local scattering clusters
\cite{poutanen2010significance}. Therefore, it is highly desirable
to exploit not only the per-link channel sparsity but also the \emph{joint
sparsity structure} to further reduce the CSIT estimation and feedback
overhead. Directly applying existing CS-based CSIT estimation for
point to point links \cite{bajwa2010compressed,berger2010application,barbotin2011estimation}
may exploit the per-link channel sparsity, but it fails to exploit
the joint sparsity in the user channel matrices. In \cite{baron2005distributed},
the authors consider a distributed CS recovery framework based on
three simple joint sparsity models. A set of jointly sparse signal
ensembles are measured distributively and recovered jointly. However,
the joint sparsity structure in our multi-user massive MIMO scenario
is much more complicated and is not covered by these existing models
\cite{baron2005distributed} and hence, the associated recovery algorithms
cannot be extended to our scenario. 
\item \textbf{Tradeoff analysis between the CSIT estimation quality and
the joint channel sparsity.} Besides the algorithm development challenge
above, it is also desirable to obtain design insights into how the
joint channel sparsity can affect the CSIT estimation performance.
However, in general, the performance analysis of the joint CS recovery
algorithms is very difficult \cite{chen2006theoretical}. In \cite{gribonval2008atoms,eldar2010average},
the authors analyze the support recovery probability of the simultaneous
orthogonal matching pursuit algorithm proposed for multiple measurement
vector problems \cite{chen2006theoretical}, and demonstrate the performance
benefits of exploiting the shared sparsity support. However, the analytical
approach cannot be easily extended to our scenario as the recovery
performance analysis is usually algorithm-specific. 
\end{itemize}

\begin{figure}
\begin{centering}
\includegraphics[scale=0.8]{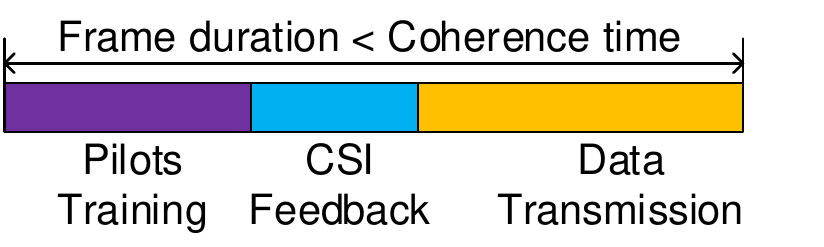}
\par\end{centering}

\caption{\label{fig:Frame-structure-with}Frame structure with pilot training
to obtain CSIT in massive MIMO FDD system.}

\end{figure}

In this paper, we propose a novel CSIT estimation and feedback method
for multi-user massive MIMO FDD systems. The proposed solution autonomously
exploits the \emph{hidden joint channel sparsity} among the users
surrounded by some \emph{common local scattering clusters} \cite{poutanen2010significance}
to substantially reduce the CSIT estimation and feedback overhead.
We first propose a joint sparsity model to incorporate the sparsity
features of the channel matrices in multi-user massive MIMO systems.
Based on this model, a distributed compressive CSIT estimation and
feedback framework is then developed such that the users obtain compressed
channel observations locally and feedback these compressed measurements
to the BS. CSIT reconstruction is performed at the BS using a joint
recovery algorithm based on the fedback compressed measurements. Specifically,
we propose, in Section III, a joint orthogonal matching pursuit algorithm
to exploit the joint channel sparsity in the CSIT recovery at the
BS. We analyze the normalized mean absolute error of the estimated
CSI in Section IV, and from the closed-form results, we obtain important
insights regarding the role of individual and distributed joint channel
sparsity in multi-user massive MIMO systems. Numerical results in
Section V demonstrate that the proposed CSIT recovery algorithm can
achieve substantial performance gains over conventional LS-based \cite{biguesh2006training,yin2013coordinated}
or existing CS-based per-link CSIT estimation solutions \cite{bajwa2010compressed,berger2010application,barbotin2011estimation}. 

\textit{Notation}s: Uppercase and lowercase boldface denote matrices
and vectors respectively. The operators $(\cdot)^{T}$, $(\cdot)^{*}$,
$(\cdot)^{H}$, $(\cdot)^{\dagger}$, $|\cdot|$, $I_{\{\cdot\}}$,
$\left\lfloor \cdot\right\rfloor $, $\left\lceil \cdot\right\rceil $,
$O(\cdot)$, $o(\cdot)$ and $\textrm{Pr}(\cdot)$ are the transpose,
conjugate, conjugate transpose, Moore-Penrose pseudoinverse, cardinality,
indicator function, round down integer, round up integer, big-O notation,
little-o notation, and probability operator respectively; $\textrm{supp}(\mathbf{h})$
is the index set of the non-zero entries of vector $\mathbf{h}$;
$\mathbf{A}(l)$, $\mathbf{a}(l)$ and $\mathbf{a}_{\Omega}$ denotes
the $l$-th column vector of $\mathbf{A}$, $l$-th entry of $\mathbf{a}$,
and sub-vector formed by collecting the entries of $\mathbf{a}$ whose
indexes are in set $\Omega$ respectively; $\mathbf{A}_{\Omega}$
and $\mathbf{A}^{\Omega}$ denote the sub-matrices formed by collecting
the columns and rows, respectively, of $\mathbf{A}$ whose indexes
are in set $\Omega$; $||\mathbf{A}||_{F}$, $||\mathbf{A}||$ and
$||\mathbf{a}||$ denote the Frobenius norm, spectrum norm of $\mathbf{A}$
and Euclidean norm of vector $\mathbf{a}$ respectively; denote $[M]=\{1,2,\cdots,M\}$.

\section{System Model}

\subsection{Multi-user Massive MIMO System}

Consider a flat block-fading multi-user massive MIMO system operating
in FDD mode. There is one BS and $K$ users in the network as illustrated
in Figure \ref{fig:Illustration-of-joint-sparsity}, where the BS
has $M$ antennas ($M$ is large) and each user has $N$ antennas.
There is a common downlink pilot channel from the BS in which the
BS broadcasts a sequence of $T$ training pilot symbols on its $M$
antennas, as illustrated in Figure \ref{fig:Illustration-of-joint-sparsity}.
Denote the transmitted pilot signal%
\footnote{We discuss in Section III-C how to choose the pilot symbols $\{\mathbf{x}_{j}\}$.%
} from the BS in the $j$-th time slot as $\mathbf{x}_{j}\in\mathbb{C}^{M\times1}$,
$j=1,\cdots,T$. The received signal vector at the $i$-th user in
the $j$-th time slot $\mathbf{y}_{ij}\in\mathbb{C}^{N\times1}$ can
be expressed as
\begin{equation}
\mathbf{y}_{ij}=\mathbf{H}_{i}\mathbf{x}_{j}+\mathbf{n}_{ij},j=1,\cdots,T,\label{eq:received_signal_model}
\end{equation}
where $\mathbf{H}_{i}\in\mathbb{C}^{N\times M}$ is the quasi-static
channel matrix from the BS to the $i$-th user and $\mathbf{n}_{ij}\in\mathbb{C}^{N\times1}$
is the complex Gaussian noise with zero mean and unit variance. Let
$\mathbf{X}=\left[\begin{array}{ccc}
\mathbf{x}_{1} & \cdots & \mathbf{x}_{T}\end{array}\right]\in\mathbb{C}^{M\times T}$, $\mathbf{Y}_{i}=\left[\begin{array}{ccc}
\mathbf{y}_{1} & \cdots & \mathbf{y}_{T}\end{array}\right]\in\mathbb{C}^{N\times T}$ and $\mathbf{N}_{i}=\left[\begin{array}{ccc}
\mathbf{n}_{i1} & \cdots & \mathbf{n}_{iT}\end{array}\right]\in\mathbb{C}^{N\times T}$ be the concatenated transmitted pilots, received signal, and noise
vectors respectively; the signal model (\ref{eq:received_signal_model})
can be equivalently written as
\begin{equation}
\mathbf{Y}_{i}=\mathbf{H}_{i}\mathbf{X}+\mathbf{N}_{i},\label{eq:revised_signal_model}
\end{equation}
where $\textrm{tr}(\mathbf{X}\mathbf{X}^{H})=PT$ is the sum transmit
SNR in the $T$ training time slots and $P$ is the transmit SNR per
time slot at the BS.

In order to effectively exploit the array gain and the spatial degrees
of freedom in massive MIMO, it is important to obtain CSIT at the
BS \cite{shi2011iteratively,bogale2012weighted}. For FDD systems,
the CSIT knowledge at the BS is obtained by two steps: (i) local CSI
estimation of $\mathbf{H}_{i}$ at the $i$-th user and (ii) feedback
of the estimated CSI $\mathbf{\hat{H}}_{i}$ to the BS, as illustrated
in Figure \ref{fig:Frame-structure-with}. Using conventional LS-based
CSI estimation techniques \cite{biguesh2006training,scaglione2004turbo},
the LS channel estimate $\mathbf{\hat{H}}_{i}$ is given by
\begin{equation}
\mathbf{\hat{H}}_{i}=\mathbf{Y}_{i}\mathbf{X}^{\dagger},\label{eq:least_square}
\end{equation}
where $\mathbf{X}^{\dagger}=\mathbf{X}^{H}\left(\mathbf{X}\mathbf{X}^{H}\right)^{-1}$
is the Moore-Penrose pseudoinverse and $\mathbf{Y}_{i}$ is the noisy
observations of the pilot symbols at the $i$-th user, as in (\ref{eq:revised_signal_model}).
After that, the estimated $\mathbf{\hat{H}}_{i}$ is then fed back
to the BS side via the reverse links. However, this LS-based approach
requires that $T\geq M$, which induces an overwhelming pilot training
and CSI feedback overhead for multi-user massive MIMO system, when
$M$ is large. Hence, this LS approach is not suitable for massive
MIMO systems, and it is very desirable to design more efficient schemes
that can exploit the hidden joint sparsity of the channel matrices
in the network.

\subsection{Joint Channel Sparsity Model}

\begin{figure}
\begin{centering}
\includegraphics[scale=0.42]{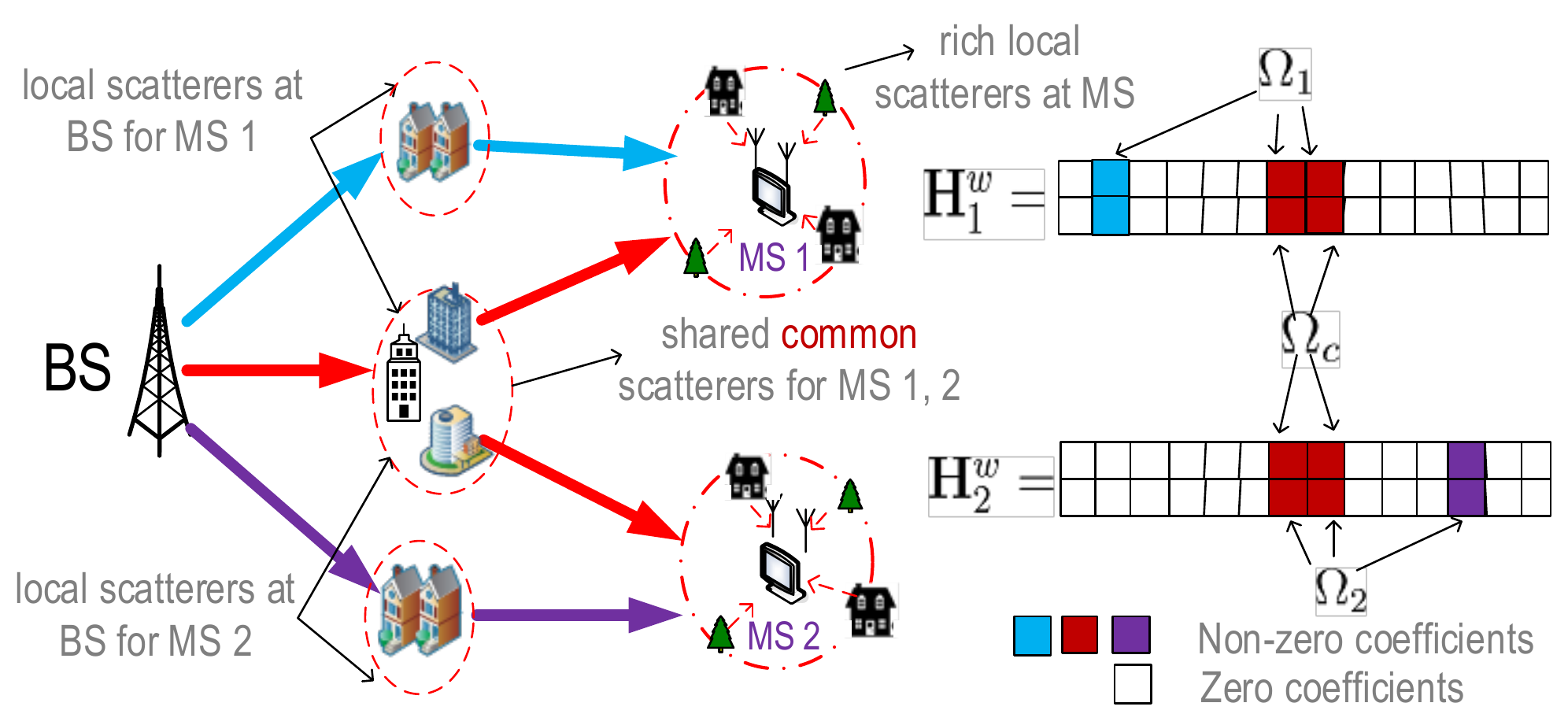}
\par\end{centering}

\caption{\label{fig:Illustration-of-joint-sparsity}Illustration of joint channel
sparsity structure due to the limited and shared local scattering
effect at the BS side. $\Omega_{c}$ is the support of the \emph{common}
scatterers shared by all users, while $\Omega_{i}$ is the support
of the individual scatterers for the $i$-th user. }
\end{figure}

We assume a uniform linear array (ULA) model for the antennas installed
at the BS and the users. Using the virtual angular domain representation
\cite{tse2005fundamentals}, the channel matrix $\mathbf{H}_{i}$
can be expressed as 
\[
\mathbf{H}_{i}=\mathbf{A}_{R}\mathbf{H}_{i}^{w}\mathbf{A}_{T}^{H},
\]
where $\mathbf{A}_{R}\in\mathbb{C}^{N\times N}$ and $\mathbf{A}_{T}\in\mathbb{C}^{M\times M}$
denote the unitary matrices for the angular domain transformation
at the user side and BS respectively, $\mathbf{H}_{i}^{w}\in\mathbb{C}^{N\times M}$
is the angular domain channel matrix, where its $(p,q)$-th entry
being non-zero indicates that there is a spatial path from the the
$q$-th transmit direction of the BS to the $p$-th receive direction
of the $i$-th user (see detailed physical explanation of this model
in \cite{tse2005fundamentals}). In multi-user massive MIMO systems,
as indicated by experimental study \cite{zhou2006experimental}, the
angular domain channel matrices $\{\mathbf{H}_{i}^{w}\}$ are in general
sparse due to the limited local scattering effects at the BS side.
Denote the $j$-th row of $\mathbf{H}_{i}^{w}$ as $\mathbf{h}_{ij}$.
Denote $\textrm{supp}(\mathbf{h})$ as the index set of the non-zero
entries of vector $\mathbf{h}$, i.e., $\textrm{supp}(\mathbf{h})=\{i:\mathbf{h}(i)\neq0\}$.
Based on many practical measurements of the channel matrices of massive
MIMO systems \cite{kaltenberger2008correlation,kyritsi2003correlation,hoydis2012channel,gao2011linear},
we have the following two important observations:
\begin{itemize}
\item \textbf{Observation I }(\emph{Sparsity Support within Individual Channel
Matrix}): From \cite{kaltenberger2008correlation,kyritsi2003correlation},
the massive MIMO channels are usually correlated at the BS side but
not at the user side. This is due to the limited scattering at the
BS side, and relatively rich scattering at the users (as illustrated
in Figure \ref{fig:Illustration-of-joint-sparsity}). Hence, the row
vectors within an $\mathbf{H}_{i}^{w}$ usually have the same sparsity
support, i.e., $\textrm{supp}(\mathbf{h}_{i1})=\textrm{supp}(\mathbf{h}_{i2})=\cdots\textrm{supp}(\mathbf{h}_{iN})$.
Notice that in practical multi-user massive MIMO downlinks, the BS
is usually elevated very high, with limited scatterers (relative to
the number of antennas). Hence, the BS might only have a few active
transmit directions for each user. However, the mobile users are usually
at low elevation, such that the users have relatively rich local scatterers
isotropically.
\item \textbf{Observation II }(\emph{Partially Shared Support between Different
Channel Matrices}): From \cite{kaltenberger2008correlation,hoydis2012channel,gao2011linear},
the channel matrices of different users are usually correlated (inter-channel
correlation), especially when the users are physically close to each
other. From these results \cite{kaltenberger2008correlation,hoydis2012channel,gao2011linear},
different users tend to share some common local scatterers at the
BS \cite{poutanen2010significance} (as illustrated in Figure \ref{fig:Illustration-of-joint-sparsity})
and hence, their channel matrices $\{\mathbf{H}_{i}^{w}\}$ may have
a partially common support. Specifically, there may exist a non-zero
index set $\Omega_{c}$ of the common support, such that $\Omega_{c}\subseteq\Omega_{i}$,
for all $i$. 
\end{itemize}

Note that Observation I refers to the \emph{individual joint sparsity}
within each channel matrix, while Observation II refers to the \emph{distributed
joint sparsity} among the user channel matrices. In the special case
when each user has $N=1$ antenna, each $\mathbf{H}_{i}^{w}$ would
be reduced to a row vector and hence the individual joint sparsity
from Observation I vanishes. Based on the above Observations, we have
the following assumption on the channel matrices in multi-user massive
MIMO systems. 
\begin{definitn}
[Joint Sparse Massive MIMO Channel]\label{Structured-Sparsity-Model}The
channel matrices $\{\mathbf{H}_{i}^{w}:\forall i\}$ have the following
properties:

\textbf{(a) Individual joint sparsity due to local scattering at the
BS}: Denote $\mathbf{h}_{ij}$ as the $j$-th row vector of $\mathbf{H}_{i}^{w}$;
then $\{\mathbf{h}_{ij}:\forall j\}$ are simultaneously sparse, i.e.,
there exists an index set $\Omega_{i}$, $0<|\Omega_{i}|\ll M$, $\forall i$,
such that 
\begin{equation}
\textrm{supp}(\mathbf{h}_{i1})=\textrm{supp}(\mathbf{h}_{i2})=\cdots\textrm{supp}(\mathbf{h}_{iN})\triangleq\Omega_{i}.\label{eq:comm1}
\end{equation}

\textbf{(b) Distributed joint sparsity due to common scattering at
the BS}: Different $\{\mathbf{H}_{i}^{w}:\forall i\}$ share a common
support%
\footnote{We also call $\Omega_{i}$ the support of matrix $\mathbf{H}_{i}^{w}$
as the row vectors of $\mathbf{H}_{i}^{w}$ have the same support
$\Omega_{i}$.%
}, i.e., there exists an index set $\Omega_{c}$ such that 
\begin{equation}
\bigcap_{i=1}^{K}\Omega_{i}=\Omega_{c}.\label{eq:joint1}
\end{equation}
Furthermore, the entries of $(\mathbf{H}_{i}^{w})_{\Omega_{i}}$ are
i.i.d. complex Gaussian distributed with zero mean and unit variance,
where $(\mathbf{H}_{i}^{w})_{\Omega_{i}}$ denotes the sub matrix
formed by collecting the \emph{column }vectors of $\mathbf{H}_{i}^{w}$
whose indices belong to $\Omega_{i}$. \hfill \QED
\end{definitn}

From Definition \ref{Structured-Sparsity-Model}, the massive MIMO
channel sparsity%
\footnote{Note that the proposed scheme can also be applied to the cases when
the channel is only \emph{approximately} sparse, in which case, the
close-to-zero components of the channel are treated as noise as in
(\ref{eq:revised_signal_model}).%
} support is parametrized by $\mathcal{P}=\{\Omega_{c},\{\Omega_{i}:\forall i\}\}$,
where $\Omega_{i}$ determines the \emph{individual} sparsity support
and $\Omega_{c}$ determines the shared \emph{common} sparsity support.
When $\Omega_{c}=\Omega_{i}$, $\forall i$, this reduces to the scenario
in which all users share the same local scatterers at the BS side,
and when $\Omega_{c}=\emptyset$, this reduces to the scenario in
which no common scatterers are shared by the users. We assume there
is a statistical bound on the channel sparsity levels ($\left|\Omega_{c}\right|$,
$|\Omega_{i}|$), i.e. $\Pr\left(\Lambda\right)>1-\varepsilon$ for
some small $\varepsilon$, where event $\Lambda$ denotes 
\begin{equation}
\Lambda:\quad\left|\Omega_{c}\right|\geq s_{c},\;|\Omega_{i}|\leq s_{i},\forall i.\label{eq:support_bound}
\end{equation}
and $\mathbb{S}=\left\{ s_{c},\{s_{i}:\forall i\}\right\} $ ($s_{c}$,
$s_{i}\ll M$) refers to the statistical sparsity bounds. Note that
the BS and the MSs have no knowledge of the random channel support
realizations $\mathcal{P}$. However, we assume the statistical sparsity
bound $\mathbb{S}=\left\{ s_{c},\{s_{i}:\forall i\}\right\} $ is
available to the BS. In practice, the channel sparsity statistics
$\mathbb{S}$ depends on the large scale properties of the scattering
environment and changes slowly (over a very long timescale). Hence,
knowledge of $\mathbb{S}$ can be obtained easily based on the prior
knowledge of the propagation environment (e.g., can be acquired from
offline channel propagation measurement at the BS as in \cite{barbotin2011estimation}
or long term stochastic learning and estimation \cite{bottou2002stochastic}).
Based on the above channel model, we shall elaborate our distributed
CSIT estimation and feedback framework in the next section.

\subsection{Distributed Compressive CSIT Estimation and Feedback}

In order to overcome the issue of pilot training and feedback overhead
in multi-user massive MIMO, we shall exploit the joint channel sparsity
structure defined in Definition \ref{Structured-Sparsity-Model}.
Specifically, instead of recovering each $\mathbf{H}_{i}$ based on
the observed output symbols $\mathbf{Y}_{i}$ individually at the
$i$-th user \cite{berger2010sparse,bajwa2010compressed,berger2010application},
we shall propose a novel CSIT estimation and feedback framework in
which the compressed measurements $\{\mathbf{Y}_{i}\}$ are observed
distributively at the users, while the $\{\mathbf{H}_{i}\}$ is recovered
jointly at the BS. This novel estimation topology allows us to exploit
the distributed joint sparsity among the user channel matrices to
reduce the estimation and feedback overhead to maintain a target CSIT
estimation quality. The distributed CSIT estimation and feedback algorithm
is described in the following and is also illustrated in Figure \ref{fig:Distributed-compressive-CSIT}.

\emph{Algorithm 1 (Distributed Compressive CSIT Estimation and Feedback)}
\begin{itemize}
\item \textbf{Step 1 }\emph{(Pilot Training)}: The BS sends the compressive
training symbols $\mathbf{X}\in\mathbb{C}^{M\times T}$, with $T\ll M$. 
\item \textbf{Step 2} \emph{(Compressive Measurement and Feedback)}: The
$i$-th mobile user observes the compressed measurements $\mathbf{Y}_{i}$
from the pilot symbols given in (\ref{eq:revised_signal_model}) and
feeds back to the BS side. 
\item \textbf{Step 3} \emph{(Joint CSIT Recovery at BS)}: The BS recovers
the CSIT $\{\mathbf{H}_{1}^{e},\cdots,\mathbf{H}_{K}^{e}\}$ jointly
based on the compressed feedback $\{\mathbf{Y}_{1},\cdots,\mathbf{Y}_{K}\}$.
\hfill \QED
\end{itemize}

Obviously, the pilot training and feedback overhead in Algorithm 1
are characterized by $T$. Our goal is to exploit the hidden joint
channel sparsity in the CSIT recovery in Step 3 of Algorithm 1 to
reduce the required training and feedback overhead $T$ in multi-user
massive MIMO systems. The problem of CSIT recovery at the BS (Step
3) can be formulated as follows:
\begin{problem}
[Joint CSIT Recovery at BS]\label{Structured-Sparse-Recovery}
\begin{eqnarray}
\min_{\{\mathbf{H}_{i},\forall i\}} &  & \sum_{i=1}^{K}\left\Vert \mathbf{Y}_{i}-\mathbf{H}_{i}\mathbf{X}\right\Vert _{F}^{2}\nonumber \\
\textrm{s.t.} &  & \left\{ \mathbf{H}_{i}^{w}:\forall i\right\} \mbox{ satisfy the joint sparsity}\nonumber \\
 &  & \textrm{ model as in Definition \ref{Structured-Sparsity-Model}. }\label{eq:non-convex}
\end{eqnarray}
\hfill \QED
\end{problem}

\begin{figure}
\begin{centering}
\includegraphics[scale=0.8]{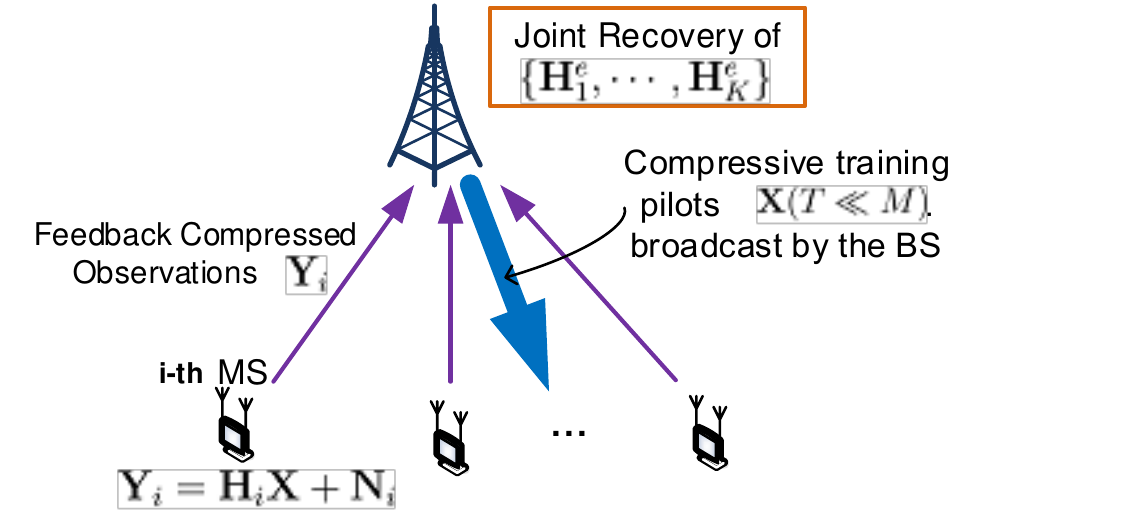}
\par\end{centering}

\caption{\label{fig:Distributed-compressive-CSIT}The BS broadcasts compressive
training pilot $\mathbf{X}\in\mathbb{C}^{M\times T}$ (with $T\ll M$)
to all the $K$ mobile users. Each user locally obtains the compressive
measurement $\mathbf{Y}_{i}$ and feeds back to the BS. The BS jointly
recovers the CSIT $\{\mathbf{H}_{1}^{e},\cdots,\mathbf{H}_{K}^{e}\}$
based on the obtained compressive measurements $\{\mathbf{Y}_{1},\cdots,\mathbf{Y}_{K}\}$. }
\end{figure}

However, Problem \ref{Structured-Sparse-Recovery} is very challenging
due to the individual and distributed joint sparsity requirement in
constraint (\ref{eq:non-convex}). This is quite different from the
conventional CS-recovery problem with a simple sparsity ($\mbox{l}_{0}$-norm)
constraint. In later sections, we shall propose a low complexity greedy
algorithm to solve Problem \ref{Structured-Sparse-Recovery}. 

\vspace{-0.5cm}

\begin{center}
\framebox{\begin{minipage}[t]{1\columnwidth}%
Challenge 1: Design a low complexity algorithm to solve Problem 1
despite challenging constraint (\ref{eq:non-convex}). %
\end{minipage}}
\par\end{center}

\section{Joint CSIT Recovery Algorithm Design}

In this section, we shall propose a low complexity algorithm to solve
Problem \ref{Structured-Sparse-Recovery} by exploiting the hidden
sparsity structures of the channel matrices (Definition 1). To achieve
this, we first rewrite (\ref{eq:revised_signal_model}) into the standard
CS model. Denote the following new variables: 
\begin{equation}
\mathbf{\bar{Y}}_{i}=\sqrt{\frac{M}{PT}}\mathbf{Y}_{i}^{H}\mathbf{A}_{R}\in\mathbb{C}^{T\times N},\;\bar{\mathbf{X}}=\sqrt{\frac{M}{PT}}\mathbf{X}^{H}\mathbf{A}_{T}\in\mathbb{C}^{T\times M},\label{eq:transformed_X}
\end{equation}
\begin{equation}
\mathbf{\bar{H}}_{i}=(\mathbf{H}_{i}^{w})^{H}\in\mathbb{C}^{M\times N},\;\mathbf{\bar{N}}_{i}=\sqrt{\frac{M}{PT}}\mathbf{N}_{i}^{H}\mathbf{A}_{R}\in\mathbb{C}^{T\times N}.\label{eq:transformed_H}
\end{equation}
Substituting these variables into (\ref{eq:revised_signal_model}),
we obtain 
\begin{equation}
\mathbf{\bar{Y}}_{i}=\bar{\mathbf{X}}\mathbf{\bar{H}}_{i}+\mathbf{\bar{N}}_{i},\forall i.\label{eq:further_revised_model}
\end{equation}
Then (\ref{eq:further_revised_model}) matches the standard CS measurement
model, where $\bar{\mathbf{X}}$ is the measurement matrix with $\textrm{tr}(\bar{\mathbf{X}}^{H}\bar{\mathbf{X}})=M$
and $\mathbf{\bar{H}}_{i}$ is the sparse matrix. Further note that
\[
\left\Vert \mathbf{\bar{Y}}_{i}-\bar{\mathbf{X}}\mathbf{\bar{H}}_{i}\right\Vert _{F}^{2}=\frac{M}{PT}\left\Vert \mathbf{Y}_{i}-\mathbf{H}_{i}\mathbf{X}\right\Vert _{F}^{2},\forall i
\]
and hence solving Problem \ref{Structured-Sparse-Recovery} is equivalent
to finding $\{\mathbf{\bar{H}}_{i}:\forall i\}$ to minimize $(\sum_{i=1}^{K}\left\Vert \mathbf{\bar{Y}}_{i}-\bar{\mathbf{X}}\mathbf{\bar{H}}_{i}\right\Vert _{F}^{2})$
subject to the joint sparsity constraint (\ref{eq:non-convex}). Based
on this equivalence relationship and equation (\ref{eq:further_revised_model}),
we elaborate our designed algorithm to solve Problem \ref{Structured-Sparse-Recovery}
in the following.

\subsection{Proposed J-OMP Algorithm}

In the literature, numerous algorithms \cite{candes2005decoding,tropp2007signal,needell2009cosamp,dai2009subspace,duarte2011structured,tropp2006algorithms,eldar2009robust,tropp2006algorithms2,liang2010joint,baraniuk2010model}
have been proposed to solve the CS problems either with or without
\emph{structure}d sparsity of the signal sources. Classical CS recovery
algorithms that consider general sparse signals without \emph{structured}
sparsity include the basis pursuit (BP) \cite{candes2005decoding},
orthogonal matching pursuit (OMP) \cite{tropp2007signal}, and many
variants of OMP such as the compressive sampling matching pursuit
(CoSaMP) \cite{needell2009cosamp} and subspace pursuit (SP) \cite{dai2009subspace}.
Based on these initial CS works \cite{candes2005decoding,tropp2007signal,needell2009cosamp,dai2009subspace},
many later works \cite{duarte2011structured,tropp2006algorithms,tropp2006algorithms2,eldar2009robust,liang2010joint,baraniuk2010model}
have considered \emph{structured} sparse signals and looked for approaches
to exploit the structured sparsity properties. For instance, in \cite{duarte2011structured,tropp2006algorithms},
a simultaneous OMP (SOMP) algorithm is proposed to solve the multiple
measurement vector (MMV) problems \cite{duarte2011structured}. In
\cite{eldar2009robust,tropp2006algorithms2}, a mixed-norm BP is developed
to recover the block sparse signals and in \cite{liang2010joint},
a select-discard OMP (SD-OMP) is proposed to recover partially joint
sparse signals. However, these joint sparsity structures \cite{duarte2011structured,tropp2006algorithms,tropp2006algorithms2,liang2010joint,eldar2009robust}
do not cover the cases of the channel matrices $\{\mathbf{H}_{i}^{w}\}$
in this paper and hence, the associated algorithms \cite{duarte2011structured,tropp2006algorithms,tropp2006algorithms2,liang2010joint,eldar2009robust}
is not suitable to solve our problem (Problem 1). In \cite{baraniuk2010model},
a model-based CS framework is further proposed to model the structured
sparse signals and some sparse approximation algorithms are proposed
to recover several special classses of structured sparse signals.
However, this framework cannot be extended to our problem because
the \emph{pruning }step in the model-based CoSaMP\cite{baraniuk2010model}
is still combinatorial under our scenario. In this section, we shall
propose a novel joint orthogonal matching pursuit (J-OMP) algorithm
to solve Problem \ref{Structured-Sparse-Recovery}. Specifically,
the proposed J-OMP algorithm is designed by extending conventional
OMP \cite{tropp2007signal} to adapt to the specific sparsity structures
of massive MIMO channels discussed in Section II. 

Denote $\mathbf{A}^{\Omega}$ as the sub matrix formed by collecting
the \emph{row vectors} of $\mathbf{A}$ whose indices belong to $\Omega$.
The details of the proposed algorithm are given below:

\emph{Algorithm 2 (Joint-OMP to Solve Problem \ref{Structured-Sparse-Recovery})}

\textbf{Input}: $\{\mathbf{Y}_{i}:\forall i\}$, $\mathbf{X}$, $\mathbb{S}=\left\{ s_{c},\{s_{i}:\forall i\}\right\} $,
$\eta_{1}$, $\eta_{2}$ $(\eta_{1}<1,\eta_{2}>1)$.

\textbf{Output}: \emph{E}stimated $\{\mathbf{\mathbf{H}}_{i}^{e}\}$
for $\{\mathbf{\mathbf{H}}_{i}:\forall i\}$.
\begin{itemize}
\item \textbf{Step 1} (\emph{Initialization}): Compute $\mathbf{\bar{Y}}_{i}$,
$\forall i$, $\bar{\mathbf{X}}$ from $\{\mathbf{Y}_{i}:\forall i\}$
and $\mathbf{X}$, as in (\ref{eq:transformed_X}).
\item \textbf{Step 2} (\emph{Common Support Identification}): Initialize
$\mathbf{R}_{i}=\mathbf{\bar{Y}}_{i}$, $\forall i$, $\Omega_{c}^{e}=\emptyset$
and then repeat the following procedures $s_{c}$ times.

\begin{itemize}
\item \textbf{A }(\emph{Support Estimate)}:\textbf{ }Estimate the remaining
index set by $\Omega_{i}^{'}=\arg\max_{\left|\Omega\right|=s_{i}-|\Omega_{c}^{e}|}\left\Vert (\bar{\mathbf{X}}_{\Omega})^{H}\mathbf{R}_{i}\right\Vert _{F}$,
$\forall i$.
\item \textbf{B} \emph{(Support Pruning): }Prune support $\Omega_{i}^{'}$
to be $\Omega_{i}^{'}=\left\{ j:\, j\in\Omega_{i}^{'},\left\Vert \bar{\mathbf{X}}(j)^{H}\mathbf{R}_{i}\right\Vert _{F}^{2}\geq\eta_{1}N\right\} $,
$\forall i$\@.
\item \textbf{C} \emph{(Support Update)}: Update the estimated\emph{ common}
support as $\Omega_{c}^{e}=\Omega_{c}^{e}\bigcup\left\{ \textrm{arg}\max_{j}\sum_{i=1}^{K}I_{\{j\in\Omega_{i}^{'}\}}\right\} .$
\item \textbf{D} \emph{(Residual Update)}: $\mathbf{R}_{i}=\left(\mathbf{I}-\mathbf{P}_{\Omega_{c}^{e}}\right)\mathbf{\bar{Y}}_{i}$,
where%
\footnote{Note that we have a slight abuse of the notation usage regarding having
$\Omega_{c}$ as a subscript of $\mathbf{P}$ for the sake of conciseness.%
} $\mathbf{P}_{\Omega_{c}}$ is a projection matrix and is given by
\begin{equation}
\mathbf{P}_{\Omega_{c}^{e}}=(\bar{\mathbf{X}}_{\Omega_{c}^{e}})(\bar{\mathbf{X}}_{\Omega_{c}^{e}})^{\dagger}.\label{eq:projection_definition}
\end{equation}

\end{itemize}

\item \textbf{Step 3} (\emph{Individual Support Identification}): Set $\Omega_{i}^{e}=\Omega_{c}^{e}$,
$\forall i$ and estimate the individual support $\Omega_{i}^{e}$
for each user $i$ individually. Specifically, for the $i$-th user,
stop if $||\mathbf{R}_{i}||_{F}^{2}\leq\frac{\eta_{2}NM}{P}$ or the
following procedures have been repeated $(s_{i}-s_{c})$ times.

\begin{itemize}
\item \textbf{A} \emph{(Support Update)}: Update the estimated \emph{individual}
support as\textbf{ }$\Omega_{i}^{e}=\Omega_{i}^{e}\bigcup\left\{ \arg\max_{j}\left\Vert \bar{\mathbf{X}}(j)^{H}\mathbf{R}_{i}\right\Vert _{F}\right\} $.
\item \textbf{B} \emph{(Residual Update): }$\mathbf{R}_{i}=\left(\mathbf{I}-\mathbf{P}_{\Omega_{i}^{e}}\right)\mathbf{\bar{Y}}_{i}$.
\end{itemize}

\item \textbf{Step 4} \emph{(Channel Estimation by LS)}: The estimated channel
for user $i$ is $\mathbf{\mathbf{H}}_{i}^{e}=\mathbf{A}_{R}(\mathbf{\mathbf{\bar{H}}}_{i}^{e})^{H}\mathbf{A}_{T}^{H}$,
where $\mathbf{\mathbf{\bar{H}}}_{i}^{e}$ is given by $\left(\mathbf{\mathbf{\bar{H}}}_{i}^{e}\right)^{\Omega_{i}^{e}}=\left(\bar{\mathbf{X}}_{\Omega_{i}^{e}}\right)^{\dagger}\mathbf{\bar{Y}}_{i}$,
$\left(\mathbf{\mathbf{\bar{H}}}_{i}^{e}\right)^{[M]\backslash\Omega_{i}^{e}}=\mathbf{0}$,
$\forall i$.\hfill \QED
\end{itemize}

\begin{figure}
\begin{centering}
\includegraphics[scale=0.6]{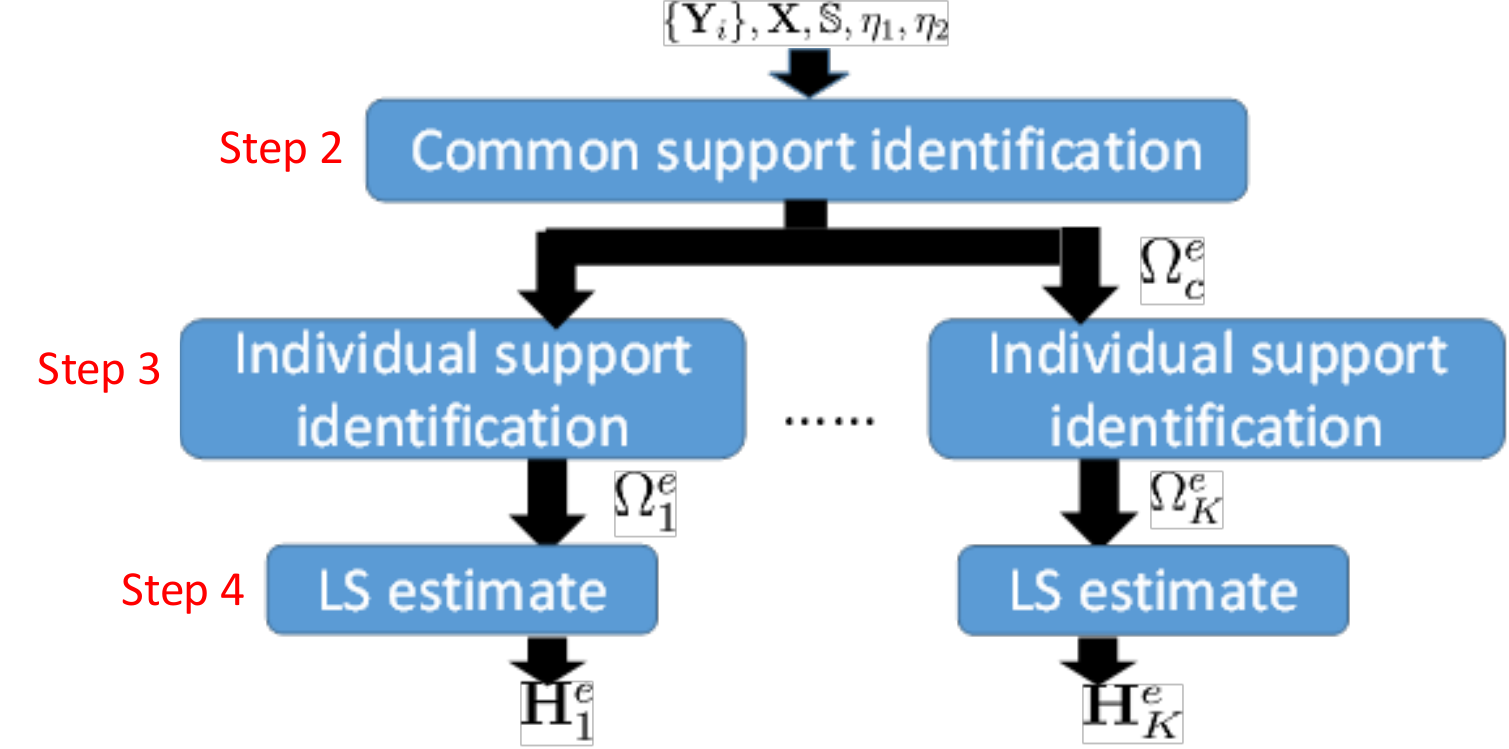}
\par\end{centering}

\caption{\label{fig:Processing-flow-of}Main processing flow of Algorithm 2.}
\end{figure}

Note that $\eta_{1}$, $\eta_{2}$ ($\eta_{1}<1$, $\eta_{2}>1$)
in the input of Algorithm 2 are threshold parameters (as in Step 2
B and Step 3). In Algorithm 2, Step 2 and 3 aim to identify the common
support $\Omega_{c}$ and the individual support $\Omega_{i}$ respectively.
Based on the estimated individual support $\Omega_{i}^{e}$, Step
4 recovers the channel matrices using the LS approach, as illustrated
in Figure \ref{fig:Processing-flow-of}. In Algorithm 2, the following
two strategies have been utilized to exploit the individual and distributed
joint sparsity in the user channel matrices.
\begin{itemize}
\item \textbf{Strategy to exploit Observation I}: Note that the estimation
target $\mathbf{\bar{H}}_{i}$ in (\ref{eq:further_revised_model})
is simultaneously zero or non-zero on each row of $N$ entries. Hence,
similar to the simultaneous recovery algorithm proposed for MMV problems
\cite{tropp2006algorithms}, we consider identifying a row vector
of $\mathbf{\bar{H}}_{i}$ \emph{an atomic unit}, based on the aggregate
matching effects between the residual $\mathbf{R}_{i}$ and the measurement
matrix $\bar{\mathbf{X}}$. For instance, we select the support index
based on the sum of the $N$ matched terms corresponding to the $N$
columns of the residual matrix $\mathbf{R}_{i}$, i.e., $\sqrt{\sum_{l=1}^{N}||\bar{\mathbf{X}}(j)^{H}\mathbf{R}_{i}(l)||^{2}}=||\bar{\mathbf{X}}(j)^{H}\mathbf{R}_{i}||_{F}$,
as in Step 2. A and Step 3. A. 
\item \textbf{Strategy to exploit Observation II}: Note that $\{\mathbf{\bar{H}}_{i}\}$
share a partial common support $\Omega_{c}$. This indicates that
the indices in $\Omega_{c}$ are very likely to be estimated by most
of the users. For instance, if $\Omega_{c}=\{i_{1}\}$, then index
$i_{1}$ is likely to be estimated as the support index by \emph{each}
of the $K$ users. Reversely speaking, the index identified by the
largest number of users is very likely to be $i_{1}$. Based on this
intuition, we have designed a joint selection process among different
users in Step 2. B, and we identify the index that appears the largest
number of times in the estimated support $\Omega_{i}^{'}$, $\forall i$,
as the next identified common support index (as in Step 2. B). 
\end{itemize}

\begin{remrk}
[Characterization of Algorithm 2]Suppose $s_{i}=s$, $\forall i$,
for simplicity. The overall complexity of Algorithm 2 is $O(KsMNT)$,
which is the same order as recovering each $\mathbf{\mathbf{\bar{H}}}_{i}$
individually using the conventional 2-norm SOMP\cite{tropp2006algorithms,candes2008introduction}.
Furthermore, compared with some conventional CS recovery algorithm
(e.g. OMP in \cite{tropp2007signal}) or sparse channel estimation
\cite{barbotin2011estimation} in which knowledge of instantaneous
sparsity level is needed, the proposed J-OMP requires only the \emph{statistical
channel sparsity information} $\mathbb{S}$, which can be estimated
using slow-timescale stochastic learning \cite{bottou2002stochastic}. 
\end{remrk}

\subsection{Analysis of Support Recovery Probability for the Proposed J-OMP Algorithm}

First of all, suppose $s_{i}$ in $\mathbb{S}$ are the same, i.e,
$s_{i}=s$, $\forall i$, to obtain simple expressions. In classical
CS works \cite{needell2009cosamp,dai2009subspace}, the restricted
isometry property (RIP) is used to characterize the measurement matrix
to facilitate the performance analysis of CS recovery algorithms.
We first review the notion of the RIP in the following: 
\begin{definitn}
[Restricted Isometry Property \cite{dai2009subspace}]\label{RIPA-matrix}Matrix
$\bar{\mathbf{X}}\in\mathbb{C}^{T\times M}$ satisfies the RIP of
order $k$ with the\emph{ restricted isometry constant} (RIC) $\delta_{k}$
if $0\leq\delta_{k}<1$ and $\delta_{k}$ is the smallest number such
that 
\[
(1-\delta_{k})\left\Vert \mathbf{h}\right\Vert ^{2}\leq\left\Vert \bar{\mathbf{X}}\mathbf{h}\right\Vert ^{2}\leq(1+\delta_{k})\left\Vert \mathbf{h}\right\Vert ^{2}
\]
holds for all $\mathbf{h}\in\mathbb{C}^{M\times1}$ where $\left\Vert \mathbf{h}\right\Vert _{0}\leq k$.
\hfill \QED
\end{definitn}

In the following analysis, we assume that the measurement matrix $\bar{\mathbf{X}}$
satisfies the RIP property%
\footnote{We will elaborate in Section III-C how to choose $\mathbf{X}$ to
make $\bar{\mathbf{X}}$ satisfy the RIP property.%
} and the specific requirements on the RICs (e.g. $\delta_{s}$, $\delta_{2s}$,
etc.) will be given in each theorem. Specifically, we are interested
in the following support recovery events%
\footnote{Note that $|\Omega_{c}|\geq s_{c}$ from (\ref{eq:support_bound}),
and hence, we can at most identify a \emph{subset} of $\Omega_{c}$
in Step 2 of Algorithm 2, i.e., $\Omega_{c}^{e}\subseteq\Omega_{c}$,
as in (\ref{eq:common_event}). %
} $\Theta_{c}$, $\{\Theta_{i}\}$ in Algorithm 2:
\begin{eqnarray}
\Theta_{c}:\textrm{ in Step 2 of Alg. 2, support }\nonumber \\
\mbox{ \ensuremath{\Omega_{c}^{e}\;}is correctly identified, i.e., \ensuremath{\Omega_{c}^{e}\subseteq\Omega_{c}}. }\label{eq:common_event}
\end{eqnarray}
\begin{eqnarray}
\Theta_{i}:\textrm{ in Step 3 of Alg. 2, support \ensuremath{\Omega_{i}^{e}}}\nonumber \\
\mbox{ is correctly identified, i.e., \ensuremath{\Omega_{i}^{e}=\Omega_{i}}. }\label{eq:individual_event}
\end{eqnarray}
We shall analyze the probability of these support recovery events
(i.e., $\Theta_{c}$, $\{\Theta_{i}\}$) because they are closely
related to the final CSIT estimation quality and the higher the support
recovery probabilities, the better the CSIT estimation quality should
be. (We formally discuss this in Theorem \ref{CSIT-Distortion-under}
in Section IV.) Denote $[M]=\{1,2,\cdots,M\}$ and suppose there is
a $K_{0}$ such that%
\footnote{The physical meaning of (\ref{eq:K_o_definition}) is that each \emph{non-common}
scatterer (i.e., corresponding to the elements outside $\Omega_{c}$)
covers no more than $K_{o}$ mobile users, as in Figure \ref{fig:Illustration-of-joint-sparsity}.%
}
\begin{equation}
\max_{j\in[M]\backslash\Omega_{c}}\sum_{i=1}^{K}I_{\{j\in\Omega_{i}\}}\leq K_{o}\label{eq:K_o_definition}
\end{equation}
where $K_{o}<K$ from (\ref{eq:joint1}). Denote $\gamma\triangleq\frac{K_{o}}{K}<1$.
Based on these preparations, we analyze the probabilities of the support
recovery events (i.e., $\Theta_{c}$, $\{\Theta_{i}\}$) by averaging
the support recovery outcomes with respect to (w.r.t.) the \emph{randomness}
of the channel matrices (Definition \ref{Structured-Sparsity-Model}).
We first have the following probability bound of event $\Theta_{c}$
conditioned on $\Lambda$ (given in (\ref{eq:support_bound})).
\begin{thm}
[Probability Bounds of $\Theta_{c}\mid\Lambda$]\label{Correct-Support-Recovery-I}If
$\theta$ and $p$ in the following satisfy {\small 
\begin{eqnarray}
\theta & \triangleq & \min\left(\frac{\left(1-2\delta_{s}\right)}{\left(\delta_{s+1}+2(1-\delta_{s})\sqrt{\frac{(1+\delta_{1})\eta_{2}M}{P}}\right)},\right.\label{eq:theta}\\
 &  & \frac{\left(1-2\delta_{s}\right)^{2}}{(1-\delta_{s})^{2}\left(\sqrt{\eta_{1}}+\sqrt{\frac{(1+\delta_{1})\eta_{2}M}{P}}\right)^{2}},\nonumber \\
 &  & \left.\frac{\left(\sqrt{\eta_{1}}-\sqrt{\frac{\eta_{2}M(1+\delta_{1})}{P}}\right)^{2}(1-\delta_{s})^{2}}{\delta_{s+1}^{2}}\right)>1,\nonumber 
\end{eqnarray}
\begin{eqnarray}
 &  & p\triangleq2\cdot\exp\left(-N\left(\ln\theta-1+\frac{1}{\theta}\right)\right)+M\cdot\exp\left(-N\times\right.\nonumber \\
 &  & \left.\left(\theta-1-\ln\theta\right)\right)+\exp\left(-NT\left(\eta_{2}-\ln\eta_{2}-1\right)\right)<1,\label{eq:ppp}
\end{eqnarray}
}then the conditional probability of event $\Theta_{c}$ given $\Lambda$,
denoted as $\textrm{Pr}(\Theta_{c}\mid\Lambda)$, satisfies
\begin{equation}
\textrm{Pr}(\Theta_{c}\mid\Lambda)\geq1-2c_{0}\sum_{t=0}^{\left\lceil \frac{1+\gamma}{2}K\right\rceil }\left(\begin{array}{c}
K\\
t
\end{array}\right)(1-p)^{t}p^{K-t},\label{eq:p_omegac}
\end{equation}
where $\eta_{1}<1,$ $\eta_{2}>1$ are the threshold parameters in
Algorithm 2, $c_{0}=\sum_{t=0}^{s_{c}}\left(\begin{array}{c}
s\\
t
\end{array}\right)-1$, $\delta_{1}$, $\delta_{s}$ and $\delta_{s+1}$ are the $1$-th,
$s$-th and $(s+1)$-th RIC respectively of $\bar{\mathbf{X}}$.\end{thm}
\begin{proof}
See Appendix \ref{sub:Proof-of-Theorem-support}.
\end{proof}

We also have the following results regarding the conditional event
of $\Theta_{i}$ given $\Theta_{c}$ and $\Lambda$, (denoted as $\Theta_{i}\mid\Theta_{c}\Lambda$),
$\forall i$.
\begin{thm}
[Probability Bounds of $\Theta_{i}\mid\Theta_{c}\Lambda$]\label{Probability-Bounds-of-individual}Denote%
\footnote{Note that $\vartheta>1$ from $\theta>1$ in (\ref{eq:theta}).%
} $\vartheta\triangleq\frac{(1-\delta_{s})P}{4\eta_{2}M}$. If (\ref{eq:theta})
holds, then the conditional probability of event $\Theta_{i}$ given
$\Theta_{c}$ and $\Lambda$, denoted as $\textrm{Pr}(\Theta_{i}\mid\Theta_{c}\Lambda)$,
satisfies
\begin{eqnarray}
 &  & \textrm{Pr}\left(\Theta_{i}\mid\Theta_{c}\Lambda\right)\geq1-c_{i}\cdot\exp\left(-N\left(\ln\theta-1+\frac{1}{\theta}\right)\right)-\nonumber \\
 &  & c_{i}M\cdot\exp\left(-N\left(\theta-1-\ln\theta\right)\right)-s\cdot\exp\left(-N\vphantom{\left(\ln\vartheta-1+\frac{1}{\vartheta}\right)}\right.\label{eq:p_omegai}\\
 &  & \left.\times\left(\ln\vartheta-1+\frac{1}{\vartheta}\right)\right)-\exp\left(-NT\left(\eta_{2}-\ln\eta_{2}-1\right)\right),\nonumber 
\end{eqnarray}
where $c_{i}=\sum_{t=s_{c}}^{s}\left(\begin{array}{c}
s\\
t
\end{array}\right)-1$. \end{thm}
\begin{proof}
Please see Appendix \ref{sub:Proof-of-Theorem-individual}.
\end{proof}

Note that Theorem 1 and 2 bound the support recovery probabilities
in terms of the joint channel sparsity parameters ($N$, $K$, $\Omega_{c}$,
etc.). On the other hand, intuitively, the channel support recovery
probability is also closely related to the final CSIT estimation quality.
As such, Theorem \ref{Correct-Support-Recovery-I} and \ref{Probability-Bounds-of-individual}
may allow us to obtain simple insights regarding how the joint channel
sparsity can benefit the CSIT estimation performance. We will have
detailed discussions of this in Section IV.

\subsection{Discussion of the Pilot Training Matrix $\mathbf{X}$}

Under the proposed CSIT estimation scheme, one issue that remains
to be discussed is how to design the entries of the $M\times T$ pilot
training matrix $\mathbf{X}$. In the CS literature, the RIP property
\cite{candes2005decoding} is widely used to characterize the quality
of a CS measurement matrix. It is shown that efficient and robust
CS recovery can be achieved when the measurement matrix satisfies
a proper RIC $\delta_{s}$ requirement (see Definition \ref{RIPA-matrix}).
On the other hand, CS measurement matrices randomly generated from
the sub-Gaussian distribution \cite{candes2008introduction} can satisfy
the RIP property with overwhelming%
\footnote{For a $T\times M$ matrix $\bar{\mathbf{X}}$ i.i.d. generated from
sub-Gaussian distribution, it is shown that when $T=c_{1}s\log M$,
the probability that $\bar{\mathbf{X}}$ fails to satisfy the $s$-th
RIP with $\delta_{s}$ decays exponentially w.r.t. $T$ as $O(\textrm{exp}(-c_{2}T)$,
where $c_{1}$ and $c_{2}$ are positive constants depending on $\delta_{s}$
\cite{candes2008introduction}.%
} probability \cite{candes2008introduction}. As such, this generation
approach for the CS measurement matrix is also widely used \cite{bajwa2010compressed,berger2010application}.
Among the sub-Gaussian family, the Rademacher distribution is commonly
adopted due to its simplicity, and it is also frequently utilized
to generate the pilot training matrix in conventional CS-based channel
estimation designs \cite{bajwa2010compressed,berger2010application}.
Based on this and from the signal model in (\ref{eq:revised_signal_model}),
the pilot training matrix $\mathbf{X}\in\mathbb{C}^{M\times T}$ can
be designed as $\mathbf{X}=\mathbf{A}_{T}\mathbf{X}_{a}$, where $\mathbf{X}_{a}\in\mathbb{C}^{M\times T}$
is i.i.d. drawn from $\left\{ -\sqrt{\frac{P}{M}},\sqrt{\frac{P}{M}}\right\} $,
with equal probability.

\section{CSIT Estimation Quality in Multi-user Massive MIMO Systems}

In this section, we analyze the CSIT estimation performance in terms
of the normalized mean absolute error \cite{sideratos2007advanced}
of the channel matrices. From the closed-form results, we can obtain
simple insights into how the joint sparsity of the user channel matrices
can be exploited to enhance the CSIT estimation performance. 

\vspace{-0.5cm}

\begin{center}
\framebox{\begin{minipage}[t]{1\columnwidth}%
Challenge 2: Derive closed-form tradeoff analysis to obtain insights
into how the \emph{joint channel sparsity} can be utilized to enhance
the CSIT estimation performance. %
\end{minipage}}
\par\end{center}

\subsection{Estimation Error of the CSIT in Multi-user Massive MIMO}

By analyzing the CSI distortion under both correct and incorrect channel
support recoveries, we obtain the following theorem on the normalized
mean absolute error (NMAE) \cite{sideratos2007advanced} of the channel
matrices. Note that we still assume $s_{i}=s$, $\forall i$, in the
statistical information $\mathbb{S}$ to obtain simple expressions. 
\begin{thm}
[CSIT Estimation Quality]\label{CSIT-Distortion-under}The normalized
mean absolute error (NMAE) \cite{sideratos2007advanced} of $\mathbf{H}_{i}$
satisfies{\small 
\begin{eqnarray}
\mathbb{E}\left(\frac{\left\Vert \mathbf{H}_{i}-\mathbf{H}_{i}^{e}\right\Vert _{F}}{\left\Vert \mathbf{H}_{i}\right\Vert _{F}}\right) & \leq & \sqrt{\frac{MNs}{PT(1-\delta_{s})}}\frac{\Gamma\left(N-\frac{1}{2}\right)}{\Gamma\left(N\right)}+\label{eq:Distortion}\\
 &  & C_{i}+E_{i}+\varepsilon\left(1+\sqrt{\frac{1+\delta_{1}}{1-\delta_{s}}}\right)\nonumber 
\end{eqnarray}
}where $\Gamma(\cdot)$ is the gamma function and 
\[
C_{i}=\left(1-\textrm{Pr}(\Theta_{c}\mid\Lambda)\right)\left(\frac{1-\delta_{s}+\delta_{2s}}{1-\delta_{s}}\right),
\]
\[
E_{i}=\left(1-\textrm{Pr}\left(\Theta_{i}\mid\Theta_{c}\Lambda\right)\right)\left(\frac{1-\delta_{s}+\delta_{2s}}{1-\delta_{s}}\right).
\]
where $\textrm{Pr}(\Theta_{c}\mid\Lambda)$ and $\textrm{Pr}\left(\Theta_{i}\mid\Theta_{c}\Lambda\right)$
are the probabilities of the conditional events $\Theta_{c}\mid\Lambda$
and $\Theta_{i}\mid\Theta_{c}\Lambda$ respectively, and $\delta_{1}$,
$\delta_{s}$ and $\delta_{2s}$ are the 1-th, $s$-th and $2s$-th
RIC of $\bar{\mathbf{X}}$ respectively. \end{thm}
\begin{proof}
See Appendix \ref{sub:Proof-of-Theorem-distortion}.
\end{proof}

Note that the term{\small{} $\varepsilon\left(1+\sqrt{\frac{1+\delta_{1}}{1-\delta_{s}}}\right)$}
in (\ref{eq:Distortion}) comes from statistical sparsity model in
(\ref{eq:support_bound}), and the terms $C_{i}$ and $E_{i}$ in
(\ref{eq:Distortion}) come from the recovery distortion in Step 2
(common support recovery) and Step 3 (individual support recovery)
of Algorithm 2 respectively. From (\ref{eq:Distortion}), a larger
support recovery probabilities of $\textrm{Pr}(\Theta_{c}\mid\Lambda)$
and $\textrm{Pr}\left(\Theta_{i}\mid\Theta_{c}\Lambda\right)$ lead
to a smaller $C_{i}$ and $E_{i}$ and hence, tend to have smaller
CSIT distortion. As such, Theorem \ref{CSIT-Distortion-under} connects
the support recovery probability results in Theorem 1 and 2 with the
final CSIT estimation quality. Based on Theorem \ref{Correct-Support-Recovery-I},
\ref{Probability-Bounds-of-individual} and \ref{CSIT-Distortion-under},
we discuss how the CSIT estimation quality is affected by the joint
channel sparsity in the next section.

\subsection{CSIT Estimation Quality w.r.t. Joint Channel Sparsity}

Recall that in Section II, we observe (i) that the channel matrix
$\mathbf{H}_{i}^{w}$ is simultaneously zero or non-zero on its columns,
with dimension $N\times1$, and (ii) that the $K$ users share a partial
common channel support $\Omega_{c}$. Therefore, it is interesting
to see how the CSIT estimation quality is affected by these joint
sparsity parameters ($N$, $K$, $\Omega_{c}$). Based on Theorem
\ref{Correct-Support-Recovery-I}, \ref{Probability-Bounds-of-individual}
and \ref{CSIT-Distortion-under}, we obtain the following results
(Corollary 1-3).
\begin{cor}
[CSIT Quality w.r.t. $N$]\label{Asymptotic-Probability-Bounds-I-1}If
$\theta$ and $p$, given in (\ref{eq:theta}), (\ref{eq:ppp}), satisfy
$\theta>1$ and $p<1$, then
\begin{equation}
\lim_{N\rightarrow\infty}-\frac{1}{N}\ln\left(C_{i}\right)\geq\left\lfloor \frac{1-\gamma}{2}K\right\rfloor \beta_{1},\label{eq:R_d_definition}
\end{equation}
\begin{equation}
\lim_{N\rightarrow\infty}-\frac{1}{N}\ln\left(E_{i}\right)\geq\beta_{2},\quad\forall i,\label{eq:R_d_definition-1}
\end{equation}
where $\beta_{1}\triangleq\min\left(\ln\theta-1+\frac{1}{\theta},\,\theta-1-\ln\theta,\, T\left(\eta_{2}-\ln\eta_{2}\right.\right.$
$\left.\left.-1\vphantom{\eta_{2}-\ln\eta_{2}}\right)\right)>0$,
$\beta_{2}\triangleq\min\left(\ln\theta-1+\frac{1}{\theta},\,\theta-1-\ln\theta,\,\ln\vartheta-1\right.,${\small{}
$\left.+\frac{1}{\vartheta},\; T\left(\eta_{2}-\ln\eta_{2}-1\right)\vphantom{\ln\vartheta-1+\frac{1}{\vartheta}}\right)>0$}.\end{cor}
\begin{proof}
See Appendix \ref{sub:Proof-of-Theorem-scalen}.\end{proof}
\begin{remrk}
[Interpretation of Corollary \ref{Asymptotic-Probability-Bounds-I-1}]Equation
(\ref{eq:R_d_definition}) and (\ref{eq:R_d_definition-1}) can be
re-written as 
\begin{equation}
C_{i}\leq\exp\left(-N\left\lfloor \frac{1-\gamma}{2}K\right\rfloor \beta_{1}+o(N)\right),\label{eq:scale1}
\end{equation}
\begin{equation}
E_{i}\leq\exp\left(-N\cdot\beta_{2}+o(N)\right).\label{eq:scale2}
\end{equation}
From (\ref{eq:scale1}) and (\ref{eq:scale2}), we conclude that $C_{i}$
and $E_{i}$ in (\ref{CSIT-Distortion-under}) decay at least exponentially
w.r.t. $N$ and hence, a larger $N$ turns out to have a smaller CSIT
estimation error from Theorem \ref{CSIT-Distortion-under}. This result
indicates that by simultaneously recovering each column of $\mathbf{H}_{i}^{w}$
as an atomic unit, as in Algorithm 2, we capture the individual joint
sparsity structure in the user channel matrices (corresponding to
Observation I). As such, a larger number of antennas $N$ at the users
turn out to have a better CSIT estimation performance. 
\end{remrk}

\begin{cor}
[CSIT Quality w.r.t. $K$]\label{Asymptotic-Probability-Bounds-I}If
(\ref{eq:theta}) holds and $p$ in (\ref{eq:ppp}) satisfies $p<\frac{1}{2}(1-\gamma)$,
then {\small 
\begin{eqnarray}
 &  & -\lim_{K\rightarrow\infty}\frac{1}{K}\ln\left(C_{i}\right)\geq R_{K}\triangleq\label{eq:asymptotic_i}\\
 &  & \qquad\frac{1-\gamma}{2}\ln\frac{(1-p)(1-\gamma)}{p(1+\gamma)}-\ln\frac{2(1-p)}{(1+\gamma)}>0.\nonumber 
\end{eqnarray}
}{\small \par}\end{cor}
\begin{proof}
Please see Appendix \ref{sub:Proof-of-Corollary-2}.
\end{proof}

\begin{remrk}
[Interpretation of Corollary \ref{Asymptotic-Probability-Bounds-I}]Equation
(\ref{eq:asymptotic_i}) can be re-written as 
\begin{equation}
C_{i}\leq\exp\left(-K\cdot R_{K}+o(K)\right),\label{eq:decay_2}
\end{equation}
where $R_{K}$ is given in (\ref{eq:asymptotic_i}). From (\ref{eq:decay_2}),
we conclude that upper bound term $C_{i}$ in (\ref{eq:Distortion})
decays at least exponentially w.r.t. the number of mobile users $K$
sharing the common channel support $\Omega_{c}$. This fact indicates
that the \emph{distributed} joint sparsity among the user channel
matrices (corresponding to Observation II) is indeed captured by the
proposed recovery algorithm and a larger $K$ tends to have a better
CSIT estimation performance.
\end{remrk}

\begin{cor}
[CSIT Quality w.r.t. $\Omega_{c}$]\label{CSIT-Distortion-omegac}Suppose
$\varepsilon=0$ in (\ref{eq:support_bound}). Scale the threshold
parameter $\eta_{2}$ in Algorithm 2 as $\eta_{2}=\sqrt{P}$ and let
the transmit SNR $P\rightarrow\infty$, the number of users $K\rightarrow\infty$.
If (\ref{eq:theta}) holds and $p$ in (\ref{eq:ppp}) satisfies $p<\frac{1}{2}(1-\gamma)$,
we have
\begin{equation}
\mathbb{E}\left(\frac{\left\Vert \mathbf{H}_{i}-\mathbf{H}_{i}^{e}\right\Vert _{F}}{\left\Vert \mathbf{H}_{i}\right\Vert _{F}}\right)\leq\left(\sum_{t=s_{c}}^{s}\left(\begin{array}{c}
s\\
t
\end{array}\right)-1\right)E,\label{eq:distortion}
\end{equation}
where $E=\left(\frac{1-\delta_{s}+\delta_{2s}}{1-\delta_{s}}\right)\times\left(\exp\left(-N\left(\ln\theta-1+\frac{1}{\theta}\right)\right)+\right.$
$\left.M\cdot\exp\left(-N\left(\theta-1-\ln\theta\right)\right)\right)$. \end{cor}
\begin{proof}
(Sketch) When (\ref{eq:theta}) holds and $p$ in (\ref{eq:ppp})
satisfies $p<\frac{1}{2}(1-\gamma)$, we obtain $\textrm{Pr}(\Theta_{c})\rightarrow1$
as $K\rightarrow\infty$. From Theorem \ref{CSIT-Distortion-under}
and Theorem \ref{Probability-Bounds-of-individual}, Corollary \ref{CSIT-Distortion-omegac}
is obtained.
\end{proof}

\begin{remrk}
[Interpretation of Theorem \ref{CSIT-Distortion-omegac}]From (\ref{eq:distortion}),
$\mathbb{E}\left(\frac{\left\Vert \mathbf{H}_{i}-\mathbf{H}_{i}^{e}\right\Vert _{F}}{\left\Vert \mathbf{H}_{i}\right\Vert _{F}}\right)\rightarrow0$
as $s_{c}\rightarrow s$ and hence a larger size (larger $s_{c}$,
$s_{c}\leq s$) of the common support $\Omega_{c}$ shared by the
users will have a smaller CSIT estimation error. This is because the
shared common channel support $\Omega_{c}$ is jointly estimated by
the $K$ users (as in Step 2 in Algorithm 2) and hence, is very likely
to be recovered when $K\rightarrow\infty$. Therefore, having a larger
size of the shared common support $\Omega_{c}$ for the users, will
lead to a better CSIT estimation performance in multi-user massive
MIMO systems.
\end{remrk}

\begin{remrk}
[Approximate Performance Comparison]Consider a baseline (Baseline
3 in Section V) which recovers each channel \emph{individually} using
the 2-norm SOMP \cite{tropp2006algorithms,candes2008introduction}
(without exploiting the common support). Then this baseline corresponds
to using $s_{c}=0$ in Algorithm 2, and the associated upper bound
on the CSIT NMAE is given by Corollary 3 (with $s_{c}=0$). Therefore,
when $K\rightarrow\infty$ and $s_{c}>\frac{1}{2}s$, the high SNR
NMAE ratio of the proposed scheme vs this baseline is given by $\mathcal{O}\left(2^{s\left(\mathcal{H}(\frac{s_{c}}{s})-1\right)}\right)$
(\cite{flum2006parameterized} Lemma 16.19), where $\mathcal{H}(\frac{s_{c}}{s})=-\frac{s_{c}}{s}\log\left(\frac{s_{c}}{s}\right)-(1-\frac{s_{c}}{s})\log\left(1-\frac{s_{c}}{s}\right)<1$,
$\frac{1}{2}<\frac{s_{c}}{s}\leq1$. Hence, the NMAE ratio decreases
as $s_{c}$ increases (i.e., as $s_{c}\rightarrow s$, $\mathcal{H}(\frac{s_{c}}{s})\rightarrow0$),
which indicates that a larger performance\emph{ gain} of the proposed
scheme over this baseline can be achieved as the size of the common
support increases (verified in Figure \ref{fig:CSIT-distortion-Sc}
in Section V).

\vspace{-0.2cm}

\end{remrk}

\section{Numerical Results}

In this section, we verify the performance advantages of our proposed
CSIT estimation scheme via simulation. Specifically, we compare the
proposed J-OMP\emph{ }recovery algorithm with the following state-of-the-art
baselines:
\begin{itemize}
\item \textbf{Baseline 1} \emph{(Conventional LS)}: Each $\mathbf{H}_{i}$
is recovered using the conventional LS approach \cite{biguesh2006training,yin2013coordinated},
as in (\ref{eq:least_square}). 
\item \textbf{Baseline 2} (\emph{OMP}): Each $\mathbf{H}_{i}$ is recovered
\emph{individually} using conventional OMP, which corresponds to a
naive extension of CS to CSIT estimations \cite{berger2010sparse,bajwa2010compressed,berger2010application}.
\item \textbf{Baseline 3} \emph{(2-norm SOMP)}: Each $\mathbf{H}_{i}$ is
recovered \emph{individually} using the 2-norm SOMP \cite{duarte2011structured,tropp2006algorithms}
to exploit the \emph{individual} joint channel sparsity. 
\item \textbf{Baseline 4} \emph{(M-BP)}: Each $\mathbf{H}_{i}$ is recovered
\emph{individually} using the mixed-norm basis pursuit\emph{ }(M-BP)
\cite{eldar2009robust,tropp2006algorithms2} to exploit the \emph{individual}
joint channel sparsity. 
\item \textbf{Baseline 5} \emph{(SD-OMP)}: The $\{\mathbf{H}_{i}\}$ are
\emph{jointly} recovered using the select-discard simultaneous OMP
algorithm proposed in \cite{liang2010joint}.
\item \textbf{Baseline 6} (\emph{Genie-aided LS}): This serves as an performance
upper bound scenario, in which we assume the BS knows the channel
support $\{\Omega_{i}\}$ with some genie aid and recovers the CSI
directly by using $\Omega_{i}$ in Step 4 of Algorithm 2. 
\end{itemize}

We consider a narrow band (flat fading) multi-user massive MIMO FDD
system with one BS and $K$ users, where the BS has $M$ antennas
and each user has $N$ antennas. Denote the average transmit SNR at
the BS as $P$ and the statistical information on the channel sparsity
levels as $\mathbb{S}=\{s_{c},\{s_{i}=s:\forall i\}\}$ as in Section
II. We use the 3GPP spatial channel model (SCM) \cite{3GPPchannel}
to generate the channel coefficients and we assume that the MS has
rich local scattering environment as in \cite{klessling2003mimo}.
On the other hand, the number of individual spatial paths and common
spatial paths from the BS broadside (corresponding to the channel
sparsity levels $|\Omega_{i}|$, $|\Omega_{c}|$ as in Section II)
are randomly generated as $|\Omega_{i}|\sim\mathcal{U}\left(s-2,s\right)$
and $|\Omega_{c}|\sim\mathcal{U}\left(s_{c},s_{c}+2\right)$ respectively
(note that $|\Omega_{c}|\geq s_{c}$, $|\Omega_{i}|\leq s$ as in
Section II), where $\mathcal{U}(a,b)$ denotes discrete uniform distribution
on the integers $\{a,a+1,...,b\}$. The spatial paths from the BS
broadside are assumed to have equal path loss and the angle of departures
are randomly and uniformly distributed over $[0,2\pi]$. The threshold
parameters $\eta_{1}$, $\eta_{2}$ in Algorithm 2 are set to be $\eta_{1}=0.2$,
$\eta_{2}=2$. \vspace{-0.2cm}

\subsection{CSIT Estimation Quality Versus Overhead $T$}

In Figure \ref{fig:CSIT-distortion-T}, we compare the normalized
mean squared error (NMSE) \cite{sideratos2007advanced} of the estimated
CSI versus the training and feedback overhead $T$, under the number
of BS antennas $M=160$, number of MS antennas $N=2$, number of users
$K=40$, common sparsity parameter $s_{c}=9$, individual sparsity
parameter $s=17$ and transmit SNR $P=28$ dB. From this figure, we
observe that the CSIT estimation quality increases as $T$ increases
and the proposed J-OMP algorithm achieves a substantial performance
gain over the baselines. This is because the proposed J-OMP exploits
the hidden joint sparsity among the user channel matrices to better
recover the CSI. Specifically, the performance gain of the 2-norm
SOMP and M-BP schemes over OMP demonstrates the benefits of exploiting
the individual joint sparsity within each channel matrix (Observation
I), and the performance gain of J-OMP over the 2-norm SOMP and M-BP
demonstrates the benefits of exploiting the distributed joint sparsity
among the user channel matrices (Observation II). Furthermore, we
observe that the proposed J-OMP, 2-norm SOMP, M-BP SD-OMP, and OMP
all approach the genie-aided LS scheme as $T$ increases. This is
because the channel support recovery probabilities of these schemes
all go to 1 as $T$ increases. This fact also highlights the importance
of having a higher probability of support recovery in the CSIT reconstruction. 

\begin{figure}
\begin{centering}
\includegraphics[width=3.5in]{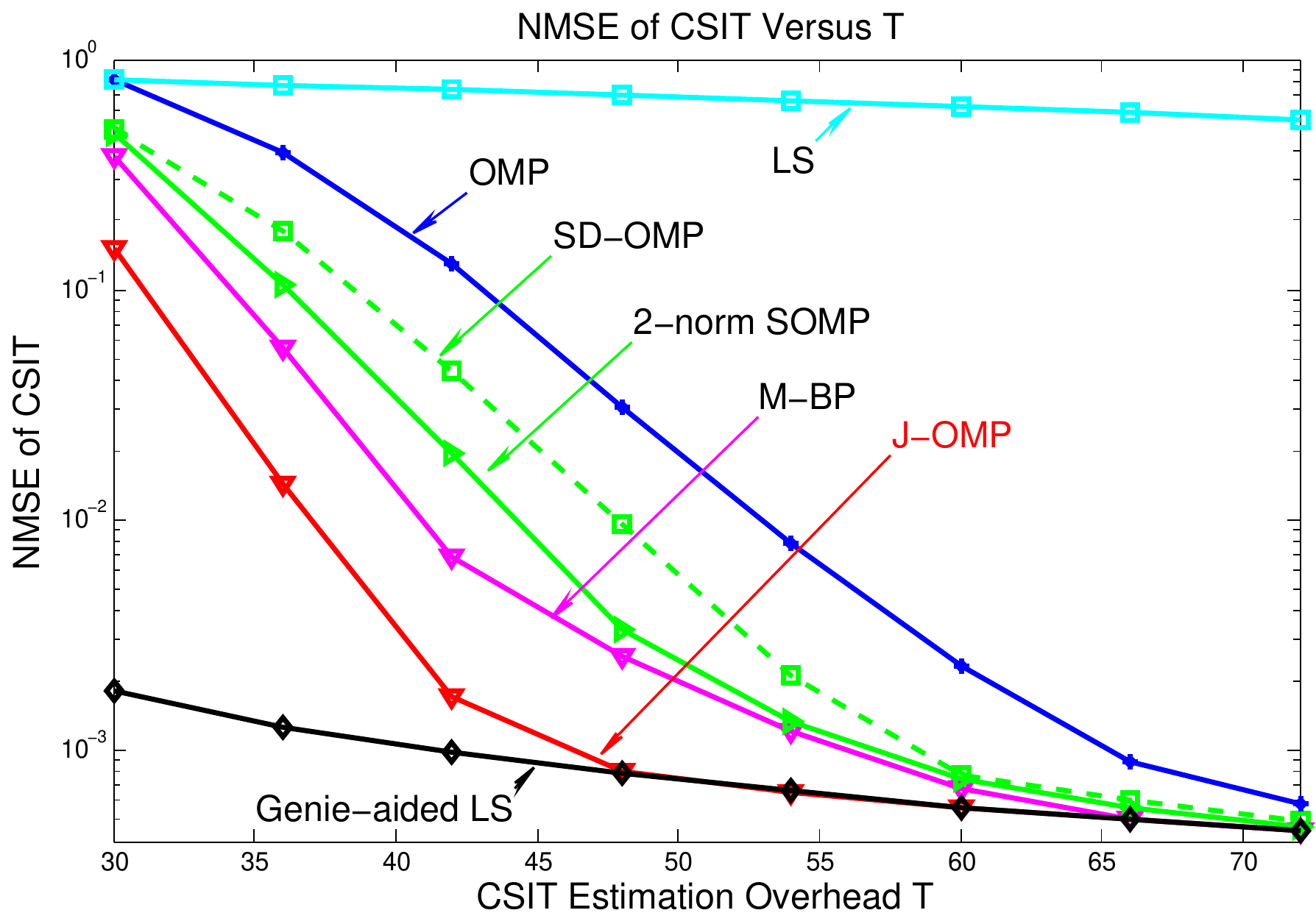}
\par\end{centering}

\caption{\label{fig:CSIT-distortion-T}NMSE of CSIT versus the CSIT training
and feedback overhead $T$ under $M=160$, $N=2$, $K=40$, $s_{c}=9$,
$s=17$ and transmit SNR $P=28$ dB.}
\end{figure}

\begin{figure}
\begin{centering}
\includegraphics[width=3.5in]{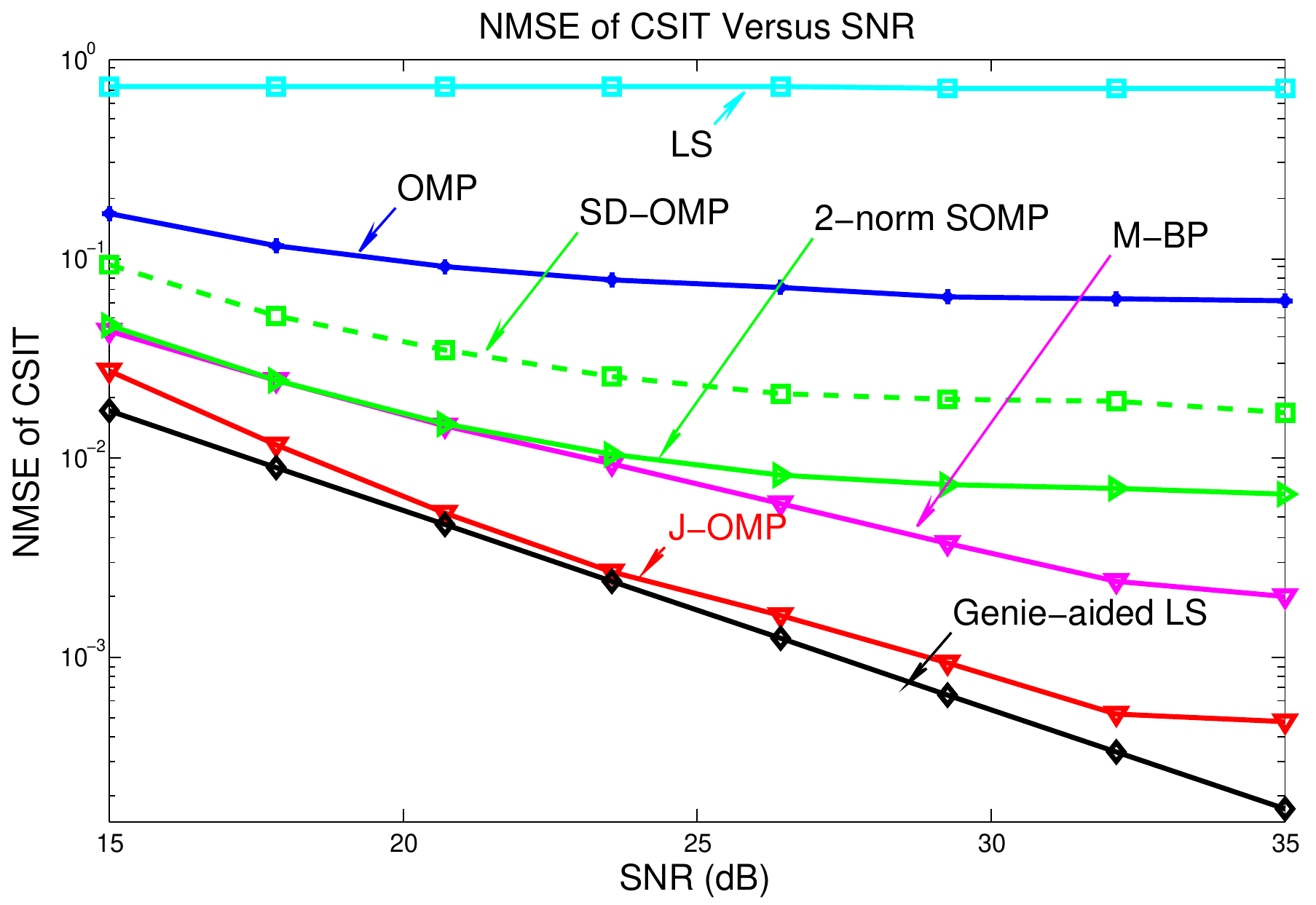}
\par\end{centering}

\caption{\label{fig:NMSE-P}NMSE of CSIT versus the transmit SNR $P$ under
under $T=45$, $M=160$, $N=2$, $K=40$, $s_{c}=9$ and $s=17$.}
\end{figure}

\begin{figure}
\begin{centering}
\includegraphics[width=3.5in]{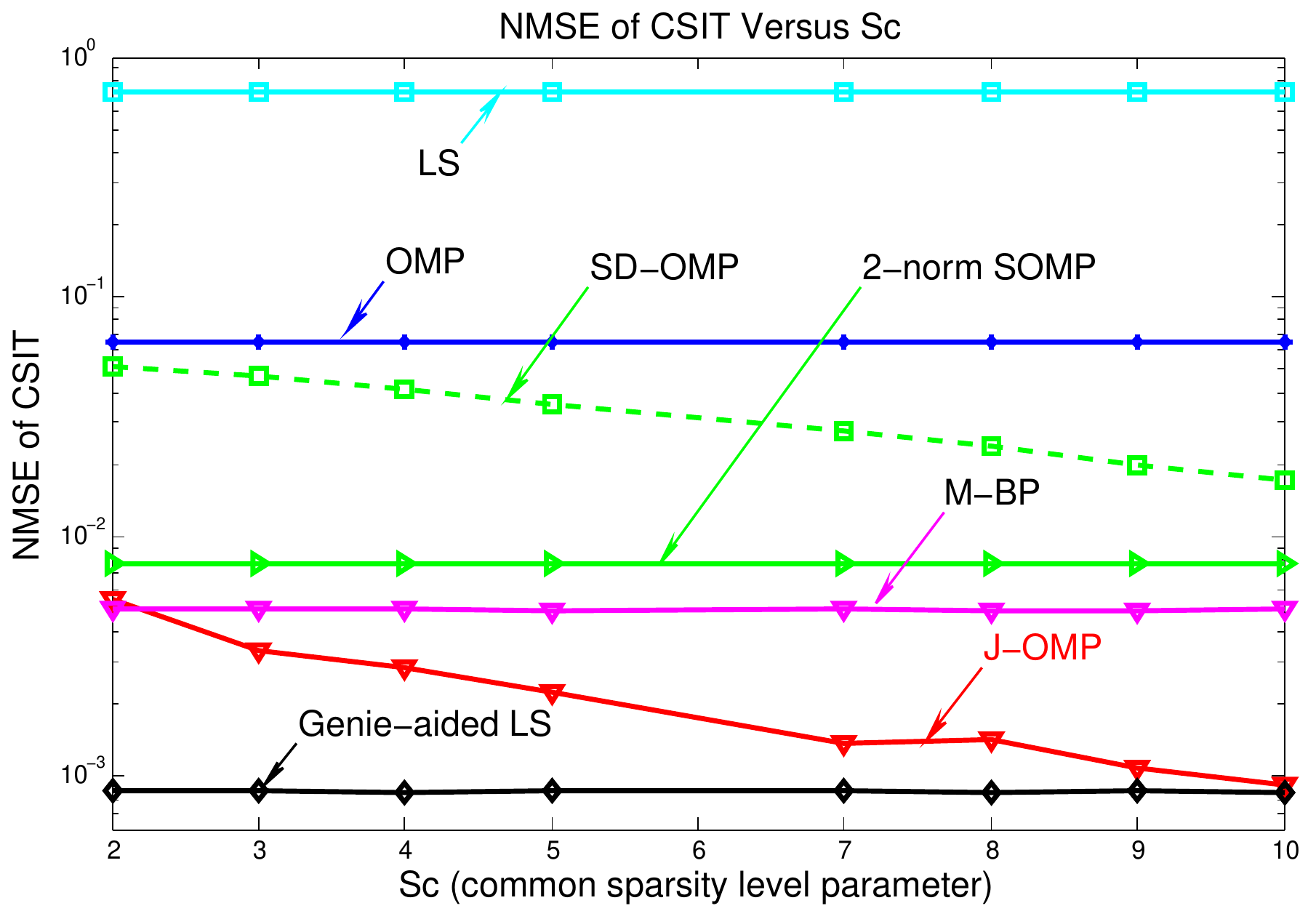}
\par\end{centering}

\caption{\label{fig:CSIT-distortion-Sc}NMSE of CSIT versus the common sparsity
$s_{c}$ under $T=45$, $M=160$, $N=2$, $K=40$, $s=17$ and transmit
SNR $P=28$ dB. }
\end{figure}

\vspace{-0.2cm}

\subsection{CSIT Estimation Quality Versus Transmit SNR $P$}

In Figure \ref{fig:NMSE-P}, we compare the NMSE of the estimated
CSI versus the transmit SNR $P$, under $T=45$, $M=160$, $N=2$,
$K=40$, $s_{c}=9$ and $s=17$. From this figure, we observe that
the proposed J-OMP algorithm has substantial performance gain over
the baselines and relatively larger performance gain is achieved in
higher SNR regions.

\subsection{CSIT Estimation Quality Versus Common Sparsity $s_{c}$}

In Figure \ref{fig:CSIT-distortion-Sc}, we compare the NMSE of the
estimated CSI versus the common sparsity level parameter $s_{c}$,
under $T=45$, $M=160$, $N=2$, $K=40$, $s=17$ and $P=28$ dB.
From this figure, we observe that the CSIT estimation quality of the
proposed J-OMP scheme gets better as the size of the common support
$s_{c}$ increases. This is because the proposed J-OMP scheme exploits
the shared common support of the $K$ users and the common support
is more likely to be correctly identified as in Algorithm 2. Hence,
as the size of the common support increase (larger $s_{c}$), better
CSIT estimation performance is achieved. This result also verifies
Corollary \ref{CSIT-Distortion-omegac} and shows that the proposed
scheme indeed exploits the \emph{distributed} joint sparsity (Observation
II) support among the users to enhance the CSIT estimation performance. 

\begin{figure}
\centering{}\includegraphics[width=3.5in]{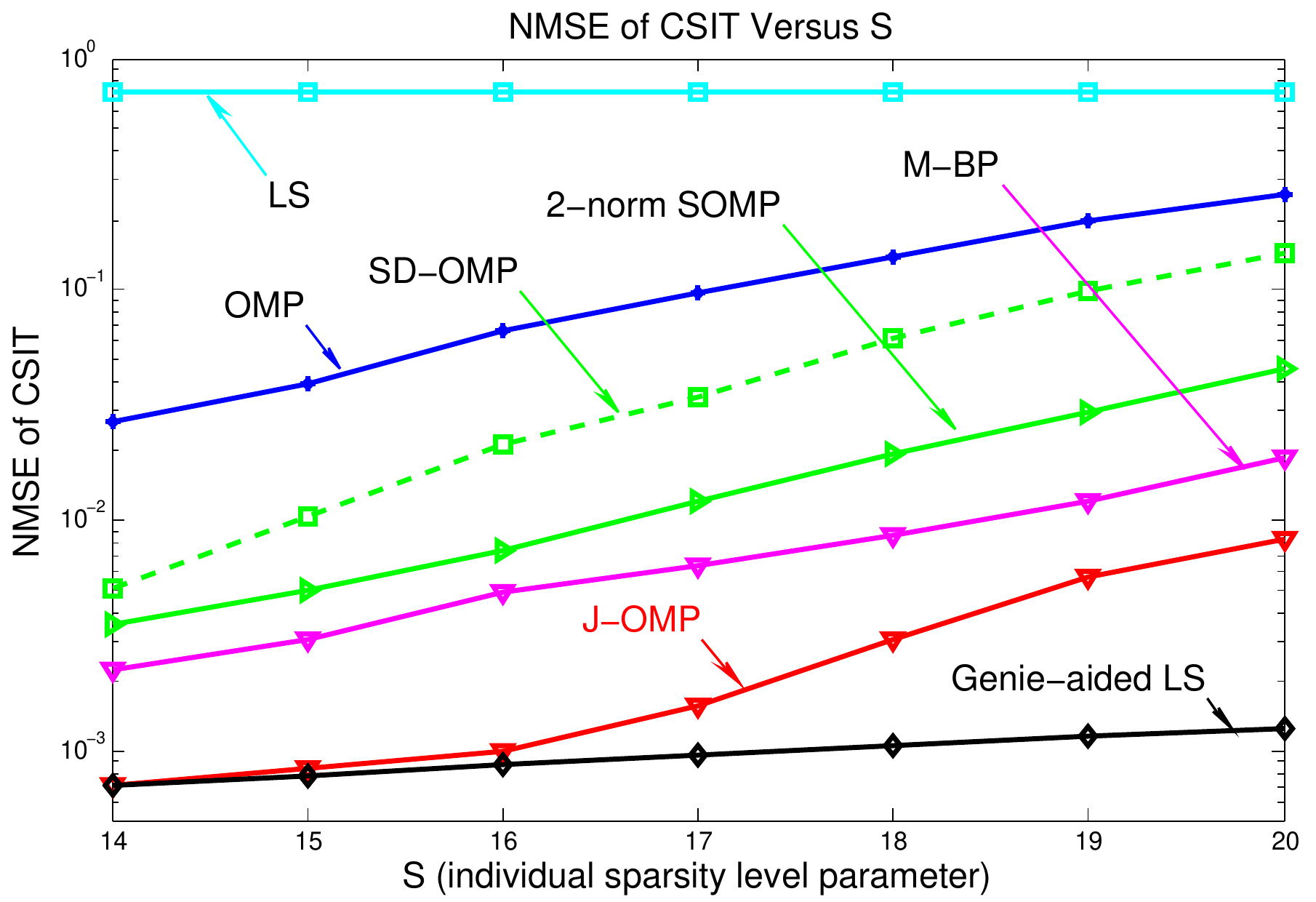}\caption{\label{fig:NMSE-S}NMSE of CSIT versus the individual sparsity $s$
under $T=45$, $M=160$, $N=2$, $K=40$, $s_{c}=9$ and transmit
SNR $P=28$ dB.}
\end{figure}

\begin{figure}
\centering{}\includegraphics[width=3.5in]{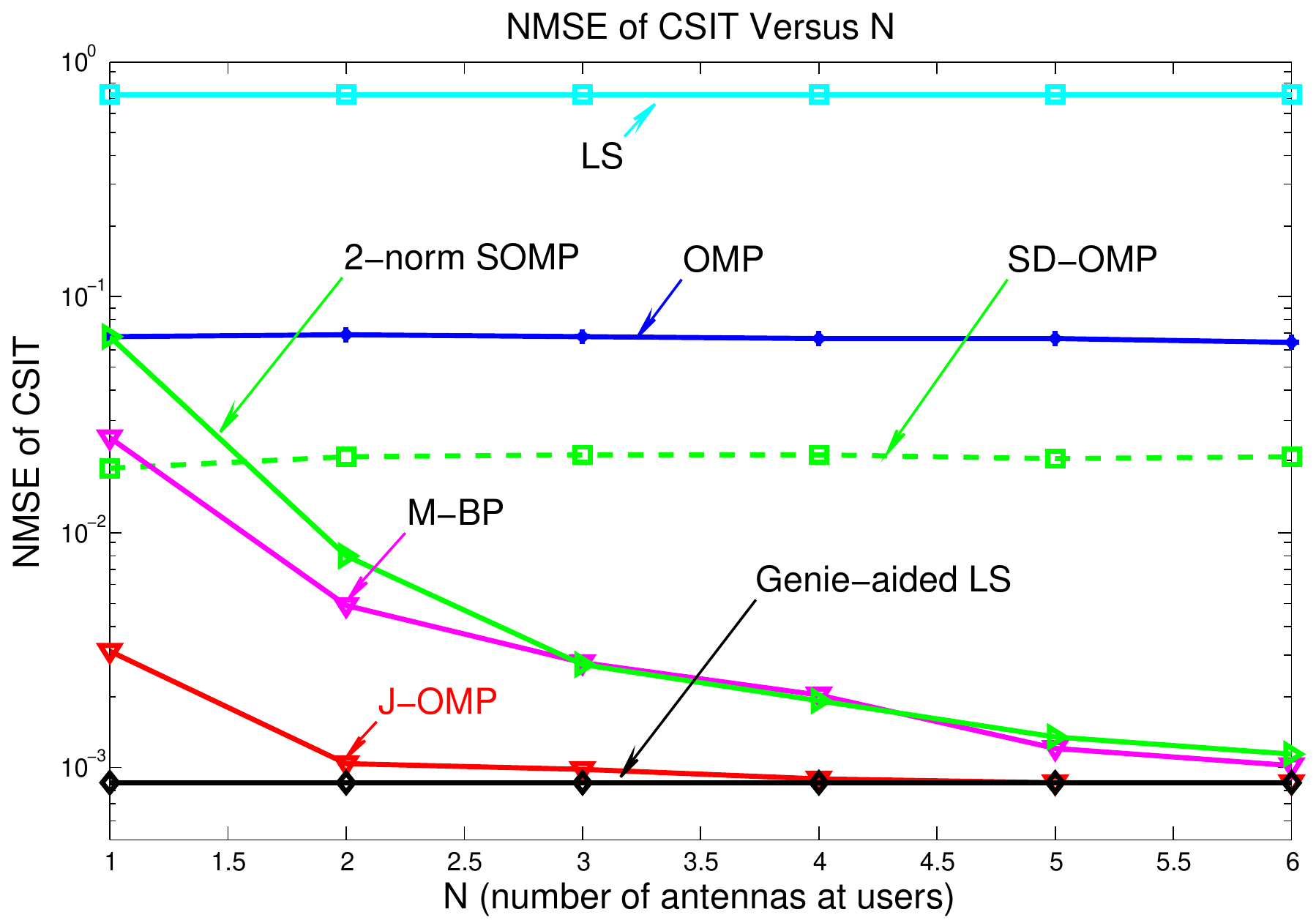}\caption{\label{fig:NMSE-N}NMSE of CSIT versus the MS antennas $N$ under
$T=45$, $M=160$, $K=40$, $s_{c}=9$, $s=17$ and transmit SNR $P=28$
dB.}
\end{figure}

\begin{figure}
\centering{}\includegraphics[width=3.5in]{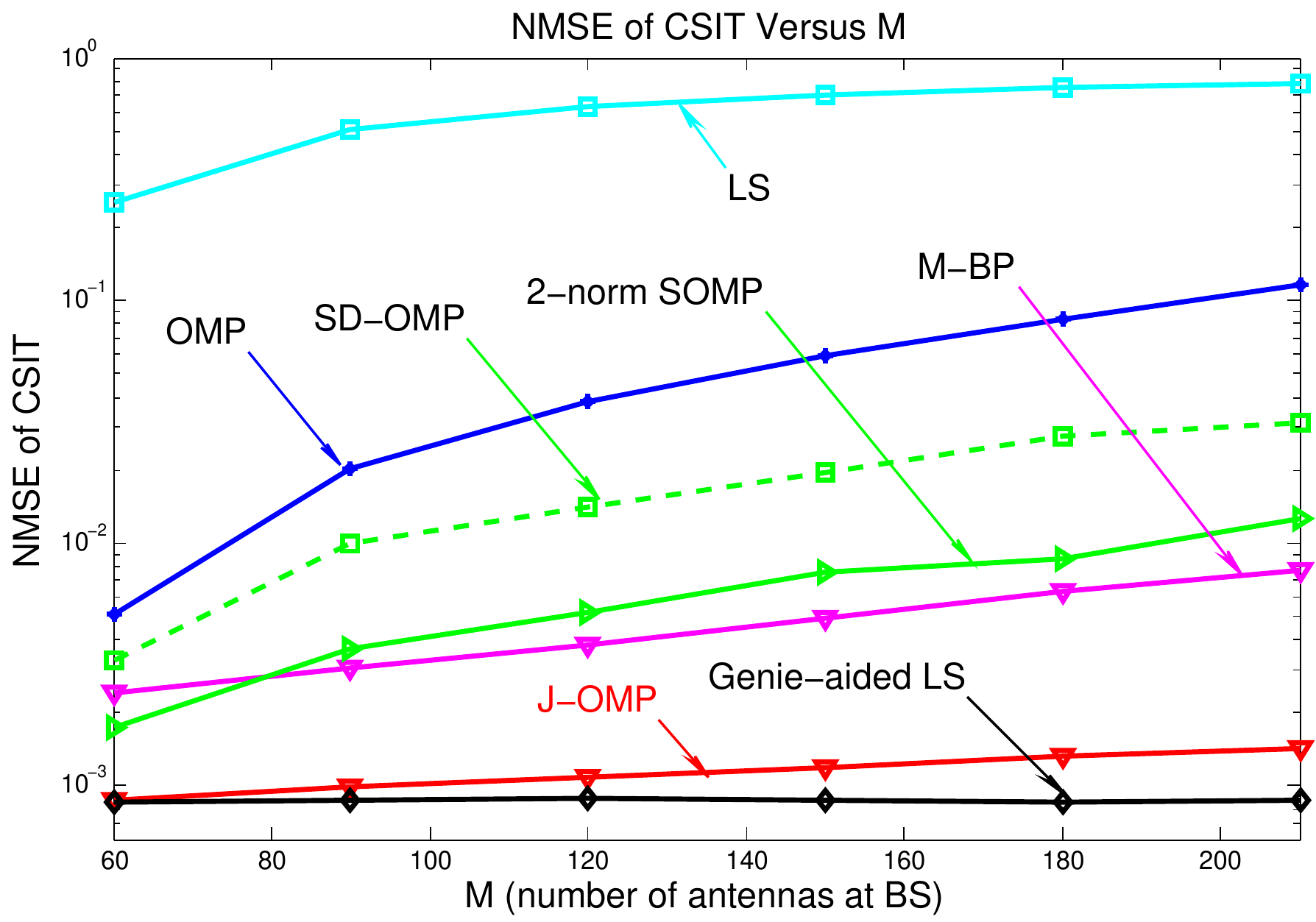}\caption{\label{fig:NMSE-M}NMSE of CSIT versus the BS antennas $M$ under
$T=45$, $N=2$, $K=40$, $s_{c}=9$, $s=17$ and transmit SNR $P=28$
dB.}
\end{figure}

\begin{figure}
\centering{}\includegraphics[width=3.5in]{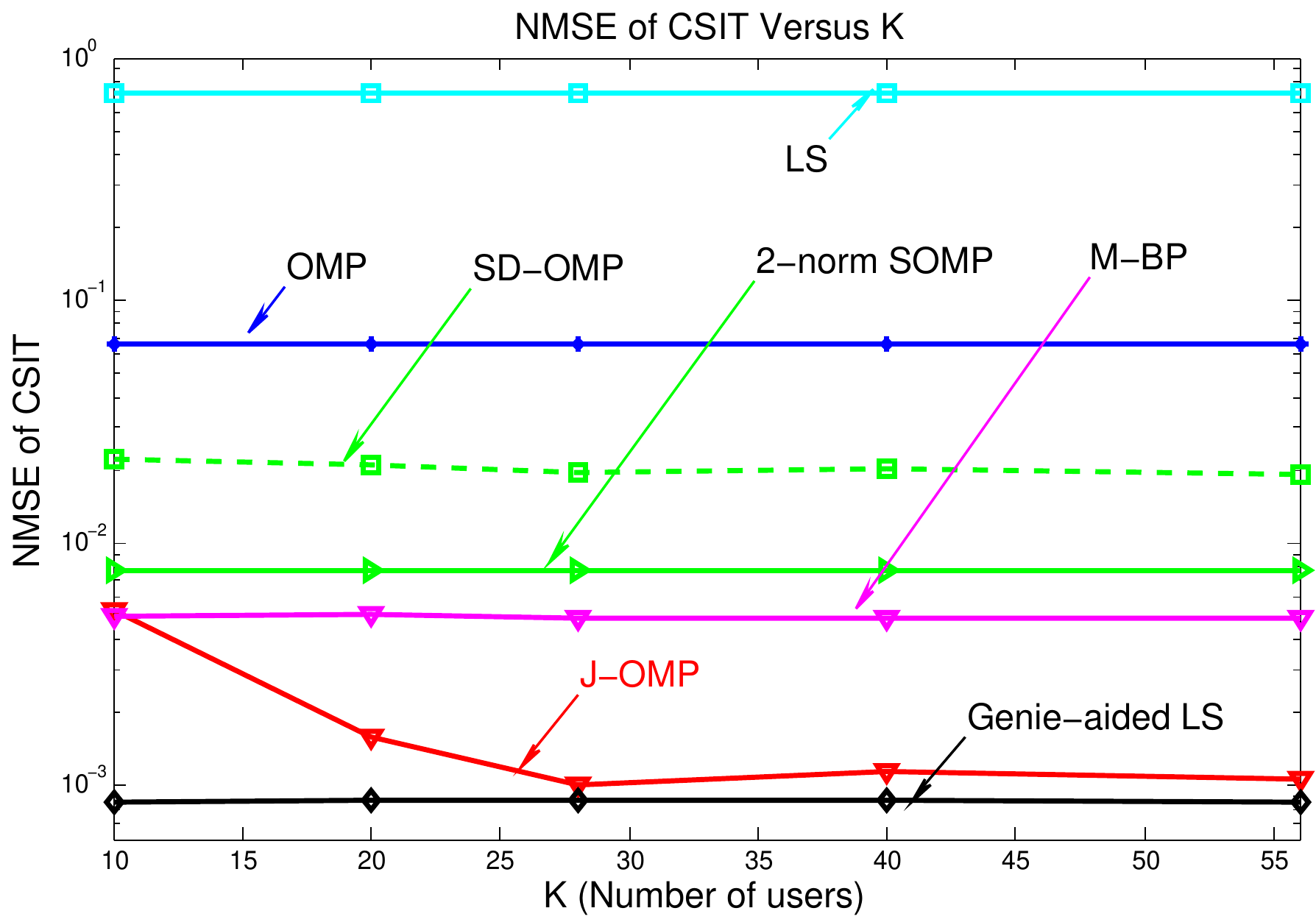}\caption{\label{fig:NMSE-K}NMSE of CSIT versus the number of MSs $K$ under
$T=45$, $M=160$, $N=2$, $s_{c}=9$, $s=17$ and transmit SNR $P=28$
dB.}
\end{figure}

\vspace{-0.5cm}

\subsection{CSIT Estimation Quality Versus Individual Sparsity $s$}

In Figure \ref{fig:NMSE-S}, we compare the NMSE of the estimated
CSI versus the individual sparsity parameter $s$, under $T=45$,
$M=160$, $N=2$, $K=40$, $s_{c}=9$ and $P=28$ dB. From this figure,
we observe that the CSIT estimation quality gets worse as the individual
sparsity level $s$ increases. This is because as the sparsity level
($s$) increases, larger number of measurements are needed according
to the classical CS theory \cite{berger2010application}. Hence, given
the same CSIT estimation overhead $T$, the CSIT estimation performance
decreases as the channel sparsity increases.

\subsection{CSIT Estimation Quality Versus MS Antennas $N$}

In Figure \ref{fig:NMSE-N}, we compare the NMSE of the estimated
CSI versus the MS antennas $N$, under $T=45$, $M=160$, $K=40$,
$s_{c}=9$, $s=17$ and $P=28$ dB. From this figure, we observe that
the CSIT estimation quality of the proposed scheme increases as $N$
increases. This is because the proposed J-OMP algorithm exploits the
individual joint sparsity among the $N$ row vectors of each channel
matrices. Hence, larger $N$ would have better CSIT estimation quality.
This result also verifies Corollary \ref{Asymptotic-Probability-Bounds-I-1}
and shows that the proposed CSIT estimation scheme indeed exploits
the \emph{individual} joint sparsity (Corresponds to Observation I)
of the user channel matrices to enhance the CSIT estimation performance.

\subsection{CSIT Estimation Quality Versus BS Antennas $M$}

In Figure \ref{fig:NMSE-M}, we compare the NMSE of the estimated
CSI versus the BS antennas $M$ under $T=45$, $N=2$, $K=40$, $s_{c}=9$,
$s=17$ and $P=28$ dB. From this figure, the proposed J-OMP algorithm
achieves better performance than the baselines. On the other hand,
we observe that the CSIT estimation quality decreases as $M$ increases.
This is because when the channel dimension (i.e., $M$) gets larger,
larger CS measurements (i.e., a larger CSIT estimation overhead $T$)
are needed from classical CS theory \cite{berger2010application}.
Conversely, given the same CSIT measurements (i.e., $T$), the CSIT
estimation quality would decrease as $M$ increases.

\subsection{CSIT Estimation Quality Versus Number of MSs $K$}

In Figure \ref{fig:NMSE-K}, we compare the NMSE of the estimated
CSI versus the number of MSs $K$ under $T=45$, $M=160$, $N=2$,
$s_{c}=9$, $s=17$ and $P=28$ dB. From this figure, we observe that
the CSIT estimation quality of the proposed J-OMP scheme increases
as $K$ increases. This is because the proposed J-OMP scheme exploits
the distributed joint sparsity among the $K$ user channel matrices
(Corresponds to the observation II) to jointly identify the common
support (as in Algorithm 2). Hence, a larger $K$ would achieve better
recovery of the common support and hence achieve a better CSIT estimation
performance. This figure also verifies the derived result in Corollary
\ref{Asymptotic-Probability-Bounds-I}.

\subsection{Comparison of Computation Complexity}

In Table \ref{tab:Comparison-of-computation}, we compare the average
computation time per link (in terms of seconds) under $T=45$, $N=2$,
$K=40$, $s_{c}=9$, $s=17$ and $P=28$ dB. Three different cases
are considered, namely the number of BS antennas $M=60$, $M=120$,
and $M=180$. From this table, we observe that 1) The M-BP (mixed-norm
basis pursuit) scheme, which is an optimization-based CS recovery
algorithm \cite{eldar2009robust,tropp2006algorithms2}, has very high
(the largest) computation complexity but it cannot beat the proposed
J-OMP scheme (as in Figure \ref{fig:CSIT-distortion-T}-\ref{fig:NMSE-K})
because the proposed J-OMP has exploited the joint sparsity structures
among the user channel matrices to better recover the CSI; 2) All
the greedy-based CS recovery schemes, including the proposed J-OMP,
OMP, SOMP, SDOMP, have similar computation complexities, while the
proposed J-OMP algorithm performs the best in Figure \ref{fig:CSIT-distortion-T}-\ref{fig:NMSE-K}. 

\begin{table}
\begin{centering}
\begin{tabular}{|l|c|c|c|}
\hline 
Computation Time (s) & $M=60$ & $M=120$ & $M=180$\tabularnewline
\hline 
\hline 
Baseline1 (LS) & $4.9\times10^{-5}$ & $6.7\times10^{-5}$ & $8.1\times10^{-5}$\tabularnewline
\hline 
Baseline2 (OMP) & $1.1\times10^{-2}$ & $1.5\times10^{-2}$ & $2.0\times10^{-2}$\tabularnewline
\hline 
Baseline3 (SOMP) & $5.8\times10^{-3}$ & $8.1\times10^{-3}$ & $1.04\times10^{-2}$\tabularnewline
\hline 
Baseline4 (M-BP) & $1.3$ & $2.6$ & $3.8$\tabularnewline
\hline 
Baseline5 (SDOMP) & $1.3\times10^{-2}$ & $1.6\times10^{-2}$ & $1.8\times10^{-2}$\tabularnewline
\hline 
Proposed J-OMP & $5.4\times10^{-3}$ & $8.6\times10^{-3}$ & $1.2\times10^{-2}$\tabularnewline
\hline 
\end{tabular}
\par\end{centering}

\caption{\label{tab:Comparison-of-computation}Comparison of computation complexity
in a multi-user massive network where $T=45$, $N=2$, $K=40$, $s_{c}=9$,
$s=17$, $P=28$ dB and $M=60$,$120$, $180$.}
\end{table}

\section{Conclusion}

In this paper, we consider multi-user massive MIMO systems and deploy
the compressive sensing (CS) technique to reduce the training as well
as the feedback overhead in the CSIT estimation. We propose a distributed
compressive CSIT estimation scheme so that the compressed measurements
are observed at the users locally, while the CSIT recovery is performed
at the base station jointly. We develop joint OMP algorithm to conduct
the CSIT reconstruction which exploits the joint sparsity in the user
channel matrices. We also analyze the estimated CSIT equality in terms
of the normalized mean absolute error, and we obtain simple insights
into how the joint channel sparsity will contribute to enhancing the
CSIT estimation quality in multi-user massive MIMO systems.

\appendix

\subsection{\label{sub:Proof-of-Theorem-support}Proof of Theorem \ref{Correct-Support-Recovery-I}}

We first have the following two lemmas (Lemma \ref{Event-Implications},
\ref{Suppose--and}).
\begin{lemma}
\label{Event-Implications}If event $\Theta_{a}$ implies event $\Theta_{b}$,
then
\begin{equation}
\textrm{Pr}(\Theta_{a})\leq\textrm{Pr}\left(\Theta_{b}\right).\label{eq:event_relationship}
\end{equation}
\end{lemma}
\begin{proof}
Event $\Theta_{a}$ implies event $\Theta_{b}$ indicates that the
event space of $\Theta_{a}$ is contained in the event space \cite{varadhan1988large}
of $\Theta_{b}$. Therefore, (\ref{eq:event_relationship}) is proved.
\end{proof}

\begin{lemma}
\label{Suppose--and}Suppose $x$ and $y$ are two random variables.
The following inequality holds for any scalar $A$:
\end{lemma}
\[
\Pr\left(x<y\right)\leq\Pr\left(x\leq A\right)+\Pr\left(y>A\right),
\]
\[
\textrm{Pr}\left(x+y\geq C\right)\leq\textrm{Pr}\left(x\geq C-A\right)+\textrm{Pr}\left(y>A\right).
\]

\begin{proof}
Note that ($x<y$) implies ($x\leq A$ or $y>A$), $(x+y\geq C)$
implies ($x\geq C-A$ or $y>A$) for any $A$. Based on this and from
Lemma \ref{Event-Implications}, Lemma \ref{Suppose--and} is proved.
\end{proof}

Based on the two lemmas, given event $\Lambda$, we then investigate
the probability of event $\Theta_{c}$ w.r.t. to the randomness of
the channel matrices for \emph{any} \emph{given} support profile $\mathcal{P}=\{\Omega_{c},\{\Omega_{i}:\forall i\}\}$.
First note that event $\Theta_{c}$ happens if and only if each of
the added indices in Step 2. B belong to $\Omega_{c}$. Denote $J$
as the current estimated common support in Step 2, then the next estimated
$\Omega_{i}^{'}(J)$ in Step 2. A, B is given by:{\small 
\begin{eqnarray}
\Omega_{i}^{'}(J) & = & \left\{ j:\left\Vert (\bar{\mathbf{X}}(j))^{H}\left(\mathbf{I}-\mathbf{P}_{J}\right)\mathbf{\bar{Y}}_{i}\right\Vert _{F}^{2}\geq\eta_{1}N\right\} \label{eq:omega_pi}\\
 &  & \bigcap\arg\max_{\left|\Omega\right|=s-|J|}\left\Vert (\bar{\mathbf{X}}_{\Omega})^{H}\left(\mathbf{I}-\mathbf{P}_{J}\right)\mathbf{\bar{Y}}_{i}\right\Vert _{F}.\nonumber 
\end{eqnarray}
}{\small \par}

For the proposed J-OMP algorithm, we have the following 3 properties
(Lemma 3-5).
\begin{lemma}
\label{lemma2}Denote $J$ as the current estimated common support
in Step 2. Then any $l\subseteq J$ will not be picked again in the
next estimated $\Omega_{i}^{'}(J)$, i.e., $l\notin\Omega_{i}^{'}(J)$.\end{lemma}
\begin{proof}
The lemma is proved by noting the following two equations: 
\begin{equation}
\bar{\mathbf{X}}(l)^{H}\left(\mathbf{I}-\mathbf{P}_{J}\right)=0,\forall l\in J,\label{eq:random_sele}
\end{equation}
\begin{equation}
\bar{\mathbf{X}}(l)^{H}\left(\mathbf{I}-\mathbf{P}_{J}\right)\bar{\mathbf{Y}}_{i}\neq0,\forall l\in[M]\backslash J,\label{eq:random_n}
\end{equation}
where (\ref{eq:random_n}) holds almost surely due to the random noise
term in $\bar{\mathbf{Y}}_{i}$.
\end{proof}

From Lemma \ref{lemma2}, we obtain the following property:
\begin{lemma}
\label{lemma3}Denote $J\subseteq\Omega_{c}$, $|J|<s_{c}$, as the
current estimated common support in Step 2. The next added index in
Step 2. B will be belonging to $\Omega_{c}$ if $\mathcal{E}_{J}$
happens:
\begin{equation}
\mathcal{E}_{J}:\;\max_{l\in\Omega_{c}\backslash J}\sum_{i=1}^{K}I_{\{l\in\Omega_{i}^{'}(J)\}}>\max_{j\in[M]\backslash\Omega_{c}}\sum_{i=1}^{K}I_{\{j\in\Omega_{i}^{'}(J)\}}.\label{eq:condition_upper}
\end{equation}
where $\Omega_{i}^{'}(J)$ is given in (\ref{eq:omega_pi}). \hfill \QED
\end{lemma}

From Lemma \ref{lemma3}, we obtain the following lemma.
\begin{lemma}
\label{lemma4}Given $\Lambda$, to ensure that $\Theta_{c}$ happens,
it suffices to require that $\mathcal{E}_{J}$ in (\ref{eq:condition_upper})
happens for all $J\subseteq\Omega_{c}$, $|J|<s_{c}$.\hfill \QED
\end{lemma}

Using the negative inverse proposition of Lemma \ref{lemma4} and
from Lemma \ref{Event-Implications}, we obtain 
\begin{equation}
\textrm{Pr}(\bar{\Theta}_{c}\mid\Lambda)\leq\textrm{Pr}\left(\bigcup_{J\subseteq\Omega_{c},|J|<s_{c}}\mathcal{\bar{E}}_{J}\right)\leq\sum_{J\subseteq\Omega_{c},|J|<s_{c}}\textrm{Pr}\left(\mathcal{\bar{E}}_{J}\right),\label{eq:condition2}
\end{equation}
where $\bar{\Theta}$ means the complement event of event $\Theta$.
For a given $J\subseteq\Omega_{c},$ $|J|<s_{c}$ and $l\in\Omega_{c}\backslash J$,
we obtain
\begin{eqnarray}
 &  & \textrm{Pr}\left(\mathcal{\bar{E}}_{J}\right)\leq\nonumber \\
 &  & \Pr\left(\sum_{i=1}^{K}I_{\{l\in\Omega_{i}^{'}(J)\}}\leq\max_{j\in[M]\backslash\Omega_{c}}\sum_{i=1}^{K}I_{\{j\in\Omega_{i}^{'}(J)\}}\right)\nonumber \\
 &  & \overset{(a_{1})}{\leq}\Pr\left(\sum_{i=1}^{K}I_{\{l\in\Omega_{i}^{'}(J)\}}\leq K_{1}\right)\nonumber \\
 &  & +\Pr\left(\max_{j\in[M]\backslash\Omega_{c}}\sum_{i=1}^{K}I_{\{j\in\Omega_{i}^{'}(J)\}}\geq K_{1}\right),\label{eq:derivation_1}
\end{eqnarray}
where $(a_{1})$ uses Lemma \ref{Suppose--and} and $K_{1}$ is a
number that will be specified later. On the other hand, note that
for any $j\in[M]\backslash\Omega_{c}$, from (\ref{eq:K_o_definition}),
\begin{eqnarray*}
\sum_{i=1}^{K}I_{\{j\in\Omega_{i}^{'}\}} & = & \sum_{i}I_{\{j\in\Omega_{i}^{'},j\in\Omega_{i}\}}+\sum_{i}I_{\{j\in\Omega_{i}^{'},j\notin\Omega_{i}\}}\\
 & \leq & K_{o}+\sum_{i}I_{\{j\in\Omega_{i}^{'},j\notin\Omega_{i}\}}.
\end{eqnarray*}
Hence, we obtain 
\begin{eqnarray}
 &  & \Pr\left(\max_{j\in[M]\backslash\Omega_{c}}\sum_{i=1}^{K}I_{\{j\in\Omega_{i}^{'}\}}\geq K_{1}\right)\nonumber \\
 & \leq & \Pr\left(\sum_{i=1}^{K}\max_{j\in[M]\backslash\Omega_{i}}I_{\{j\in\Omega_{i}^{'}\}}\geq K_{1}-K_{o}\right).\label{eq:derivation2}
\end{eqnarray}

We have the following results related to (\ref{eq:derivation_1})
and (\ref{eq:derivation2}).
\begin{lemma}
\label{Probability-Bounds-I}For a given $J\subseteq\Omega_{c},$
$|J|<s_{c}$, $l\in\Omega_{c}\backslash J$, if $\theta>1,$ then
the following probability bounds holds:
\begin{equation}
\Pr\left(I_{\{l\in\Omega_{i}^{'}(J)\}}=0\right)\leq p,\label{eq:bound_zero}
\end{equation}
\begin{equation}
\Pr\left(\max_{j\in[M]\backslash\Omega_{i}}I_{\{j\in\Omega_{i}^{'}(J)\}}=1\right)\leq p,\label{eq:bound_zero_II}
\end{equation}
where $\theta$ and $p$ are given in (\ref{eq:theta}) and (\ref{eq:ppp})
respectively.\end{lemma}
\begin{proof}
See Appendix \ref{sub:Proof-of-Lemma-bound_I}.
\end{proof}

First, $\left(\mathbf{\bar{H}}_{i}\right)^{\Omega_{i}}$, $\forall i$,
are independent of each other and hence $I_{\{l\in\Omega_{i}^{'}(J)\}}$
are independent of each other for different $i$. Second, denote $\{z_{i}\}_{i=1}^{K}$
as a series of i.i.d. Bernoulli random variables, with $\Pr(z_{i}=0)=p$
and $\Pr(z_{i}=1)=1-p$, $\forall i$. By comparing $\left\{ I_{\{l\in\Omega_{i}^{'}(J)\}}:\forall i\right\} $
with $\{z_{i}\}_{i=1}^{K}$, we obtain, 
\begin{eqnarray}
 &  & \Pr\left(\sum_{i=1}^{K}I_{\{l\in\Omega_{i}^{'}(J)\}}\leq K_{1}\right)\leq\Pr\left(\sum_{i=1}^{K}z_{i}\leq K_{1}\right).\nonumber \\
 &  & \quad=\sum_{t=0}^{K_{1}}\left(\begin{array}{c}
K\\
t
\end{array}\right)(1-p)^{t}p{}^{K-t}.\label{eq:derivation_6}
\end{eqnarray}
in (\ref{eq:derivation_1}). Similarly, we obtain
\begin{eqnarray}
 &  & \Pr\left(\sum_{i=1}^{K}\max_{j\in[M]\backslash\Omega_{i}}I_{\{j\in\Omega_{i}^{'}\}}\geq K_{1}-K_{o}\right)\nonumber \\
 & \leq & \sum_{t=0}^{K-K_{1}+K_{o}}\left(\begin{array}{c}
K\\
t
\end{array}\right)(1-p)^{t}p{}^{K-t}.\label{eq:derivation_8}
\end{eqnarray}
in (\ref{eq:derivation2}). Choosing $K_{1}$ as $K_{1}=\left\lfloor \frac{K+K_{o}}{2}\right\rfloor $
and combining (\ref{eq:derivation_1}), (\ref{eq:derivation2}), (\ref{eq:derivation_6})
and (\ref{eq:derivation_8}) together, Theorem \ref{Correct-Support-Recovery-I}
is proved.

\subsection{\label{sub:Proof-of-Lemma-bound_I}Proof of Lemma \ref{Probability-Bounds-I}}

First, we introduce the following properties (\ref{eq:property1})-(\ref{eq:property_th})
from \cite{bernstein2011matrix} which will be frequently used in
the proof:

\begin{equation}
\bar{\mathbf{X}}\mathbf{\bar{H}}_{i}=\bar{\mathbf{X}}_{J}(\mathbf{\bar{H}}_{i})^{J}+\bar{\mathbf{X}}_{\Omega_{i}\backslash J}(\mathbf{\bar{H}}_{i})^{\Omega_{i}\backslash J},\forall J\subseteq\Omega_{i}.\label{eq:property1}
\end{equation}
\begin{equation}
\left(\mathbf{I}-\mathbf{P}_{J}\right)\bar{\mathbf{X}}_{J}=\mathbf{0},\;(\bar{\mathbf{X}}_{J})^{H}\left(\mathbf{I}-\mathbf{P}_{J}\right)=\mathbf{0}.\label{eq:property2}
\end{equation}
\begin{equation}
\left\Vert \mathbf{A}\mathbf{B}\right\Vert \leq\left\Vert \mathbf{A}\right\Vert \left\Vert \mathbf{B}\right\Vert ,\;\left\Vert \mathbf{A}\mathbf{B}\right\Vert _{F}\leq\left\Vert \mathbf{A}\right\Vert \left\Vert \mathbf{B}\right\Vert _{F}.\label{eq:property_th}
\end{equation}
Second, we have the following inequalities \cite{needell2009cosamp,dai2009subspace}
on the RIP property (Definition \ref{RIPA-matrix}).
\begin{lemma}
\label{For-given-inequality}For $J_{1},J_{2}\subseteq\Omega$, $|\Omega|=k$
and $J_{1}\bigcap J_{2}=\emptyset$, 
\begin{equation}
\left\Vert (\bar{\mathbf{X}}_{J_{2}})^{H}\bar{\mathbf{X}}_{J_{1}}\right\Vert \leq\delta_{k}.\label{eq:property3}
\end{equation}
\begin{equation}
\left\Vert \left((\bar{\mathbf{X}}_{J_{1}})^{H}(\bar{\mathbf{X}}_{J_{1}})\right)^{-1}\right\Vert \leq\frac{1}{1-\delta_{k}}.\label{eq:property4}
\end{equation}
\begin{equation}
\left\Vert \bar{\mathbf{X}}_{\Omega}^{\dagger}\right\Vert \leq\frac{1}{\sqrt{1-\delta_{k}}}.\label{eq:property6}
\end{equation}

\end{lemma}

\begin{flushleft}
We prove Lemma \ref{Probability-Bounds-I} based on properties (\ref{eq:property1})-(\ref{eq:property6}).
Specifically, we first prove (\ref{eq:bound_zero}) and then prove
(\ref{eq:bound_zero_II}) in a similar way. For a given $J\subseteq\Omega_{c},$
$|J|<s_{c}$, $l\in\Omega_{c}\backslash J$, we obtain{\small 
\begin{eqnarray}
\Pr\left(I_{\{l\in\Omega_{i}^{'}(J)\}}=0\right)\overset{(b_{1})}{\leq}\left(\left\Vert \bar{\mathbf{X}}(l)^{H}\left(\mathbf{I}-\mathbf{P}_{J}\right)\mathbf{\bar{Y}}_{i}\right\Vert _{F}^{2}<\eta_{1}N\right)\label{eq:b11}\\
+\mbox{Pr}\left(\left\Vert \bar{\mathbf{X}}(l)^{H}\left(\mathbf{I}-\mathbf{P}_{J}\right)\mathbf{\bar{Y}}_{i}\right\Vert \leq\max_{j\in[M]\backslash\Omega_{i}}\left\Vert \bar{\mathbf{X}}(l)^{H}\left(\mathbf{I}-\mathbf{P}_{J}\right)\mathbf{\bar{Y}}_{i}\right\Vert \right)\nonumber 
\end{eqnarray}
}where $(b_{1})$ comes from Lemma \ref{Event-Implications}, and
the fact that if both $\left\Vert \bar{\mathbf{X}}(l)^{H}\left(\mathbf{I}-\mathbf{P}_{J}\right)\mathbf{\bar{Y}}_{i}\right\Vert _{F}\geq\sqrt{\eta_{1}N}$
and 
\[
\left\Vert \bar{\mathbf{X}}(l)^{H}\left(\mathbf{I}-\mathbf{P}_{J}\right)\mathbf{\bar{Y}}_{i}\right\Vert >\max_{j\in[M]\backslash\Omega_{i}}\left\Vert \bar{\mathbf{X}}(l)^{H}\left(\mathbf{I}-\mathbf{P}_{J}\right)\mathbf{\bar{Y}}_{i}\right\Vert 
\]
hods, then we obtain $l\in\Omega_{i}^{'}(J)$ according to (\ref{eq:omega_pi}).
On the other hand, we have the following inequality for the first
term in (\ref{eq:b11}):{\small 
\begin{eqnarray}
 &  & \left\Vert \bar{\mathbf{X}}(l)^{H}\left(\mathbf{I}-\mathbf{P}_{J}\right)\mathbf{\bar{Y}}_{i}\right\Vert =\label{eq:b12}\\
 &  & \left\Vert \bar{\mathbf{X}}(l)^{H}\left(\mathbf{I}-\mathbf{P}_{J}\right)\left(\bar{\mathbf{X}}_{J}\left(\bar{\mathbf{H}}_{i}\right)^{J}+\bar{\mathbf{X}}_{\Omega_{i}\backslash J}\left(\bar{\mathbf{H}}_{i}\right)^{\Omega_{i}\backslash J}+\mathbf{\bar{N}}_{i}\right)\right\Vert \nonumber \\
 &  & \geq\left\Vert \bar{\mathbf{X}}(l)^{H}\left(\mathbf{I}-\mathbf{P}_{J}\right)\bar{\mathbf{X}}_{\Omega_{i}\backslash J}\left(\bar{\mathbf{H}}_{i}\right)^{\Omega_{i}\backslash J}\right\Vert -\sqrt{1+\delta_{1}}\left\Vert \mathbf{\bar{N}}_{i}\right\Vert _{F}\nonumber 
\end{eqnarray}
}and the following for the second term in (\ref{eq:b11}):{\small 
\begin{eqnarray}
 &  & \left\Vert \bar{\mathbf{X}}(j)^{H}\left(\mathbf{I}-\mathbf{P}_{J}\right)\mathbf{\bar{Y}}_{i}\right\Vert \label{eq:b13}\\
 &  & \leq\left\Vert \bar{\mathbf{X}}(j)^{H}\left(\mathbf{I}-\mathbf{P}_{J}\right)\bar{\mathbf{X}}_{\Omega_{i}\backslash J}\left(\bar{\mathbf{H}}_{i}\right)^{\Omega_{i}\backslash J}\right\Vert +\sqrt{1+\delta_{1}}\left\Vert \mathbf{\bar{N}}_{i}\right\Vert _{F}.\nonumber 
\end{eqnarray}
}Based on (\ref{eq:b11})-(\ref{eq:b13}) and using Lemma \ref{Suppose--and},
we obtain{\small 
\begin{eqnarray}
 &  & \Pr\left(I_{\{l\in\Omega_{i}^{'}\}}=0\right)\leq\textrm{Pr}(\mathcal{D}_{i}^{[1]})+\textrm{Pr}(\mathcal{D}_{i}^{[2]})\label{eq:derivation_2}\\
 &  & +\textrm{Pr}(\mathcal{D}_{i}^{[3]})+\textrm{Pr}\left(\left\Vert \mathbf{\bar{N}}_{i}\right\Vert _{F}\geq\sqrt{\frac{MN\eta_{2}}{P}}\right)\nonumber 
\end{eqnarray}
}
\[
\mathcal{D}_{i}^{[1]}:\quad\left\Vert \bar{\mathbf{X}}(l)^{H}\left(\mathbf{I}-\mathbf{P}_{J}\right)\bar{\mathbf{X}}_{\Omega_{i}\backslash J}\left(\bar{\mathbf{H}}_{i}\right)^{\Omega_{i}\backslash J}\right\Vert \leq\sqrt{N}\alpha,
\]
\begin{eqnarray*}
\mathcal{D}_{i}^{[2]}:\quad\max_{j\in[M]\backslash\Omega_{i}}\left\Vert \bar{\mathbf{X}}(j)^{H}\left(\mathbf{I}-\mathbf{P}_{J}\right)\bar{\mathbf{X}}_{\Omega_{i}\backslash J}\left(\bar{\mathbf{H}}_{i}\right)^{\Omega_{i}\backslash J}\right\Vert \\
\geq\sqrt{N}\alpha-2\sqrt{\frac{\eta_{2}MN(1+\delta_{1})}{P}},
\end{eqnarray*}
\begin{eqnarray*}
\mathcal{D}_{i}^{[3]}:\quad\left\Vert \bar{\mathbf{X}}(l)^{H}\left(\mathbf{I}-\mathbf{P}_{J}\right)\bar{\mathbf{X}}_{\Omega_{i}\backslash J}\left(\bar{\mathbf{H}}_{i}\right)^{\Omega_{i}\backslash J}\right\Vert \\
\leq\sqrt{\eta_{1}N}+\sqrt{\frac{\eta_{2}MN(1+\delta_{1})}{P}},
\end{eqnarray*}
where $\eta_{2}>1$ and $\alpha$ is a scalar to be set later. Note
that $\mathbf{\bar{N}}_{i}$ is i.i.d. complex Gaussian distributed
with zero mean and variance $\frac{M}{PT}$. Therefore,
\[
\textrm{Pr}\left(\left\Vert \mathbf{\bar{N}}_{i}\right\Vert _{F}\geq\sqrt{\frac{\eta_{2}MN}{P}}\right)=\textrm{Pr}\left(\chi_{2NT}\geq2\eta_{2}NT\right)
\]
 where $\chi_{2NT}$ denotes chi-squared distributed with $2NT$ degrees
of freedom. We further have the following 2 inequalities (\ref{eq:derivation_3})
and (\ref{eq:derivation_4})) from Lemma \ref{For-given-inequality}:
(i) for a given $J\subseteq\Omega_{c}$, $|J|<|\Omega_{c}|$, $l\in\Omega_{c}\backslash J$,
we have {\small 
\begin{eqnarray}
 &  & \left\Vert \bar{\mathbf{X}}(l)^{H}\left(\mathbf{I}-\mathbf{P}_{J}\right)\bar{\mathbf{X}}_{\Omega_{i}\backslash J}\right\Vert \nonumber \\
 & \geq & \left(\left\Vert \bar{\mathbf{X}}(l)^{H}\bar{\mathbf{X}}_{\Omega_{i}\backslash J}\right\Vert -\left\Vert \bar{\mathbf{X}}(l)^{H}\bar{\mathbf{X}}{}_{J}\left(\bar{\mathbf{X}}{}_{J}^{H}\bar{\mathbf{X}}_{J}\right)^{-1}\bar{\mathbf{X}}{}_{J}^{H}\bar{\mathbf{X}}_{\Omega_{i}\backslash J}\right\Vert \right)\nonumber \\
 & \geq & \left(1-\delta_{1}-\left\Vert \bar{\mathbf{X}}(l)^{H}\bar{\mathbf{X}}{}_{J}\right\Vert \left\Vert \left(\bar{\mathbf{X}}{}_{J}^{H}\bar{\mathbf{X}}_{J}\right)^{-1}\right\Vert \left\Vert \bar{\mathbf{X}}{}_{J}^{H}\bar{\mathbf{X}}_{\Omega_{i}\backslash J}\right\Vert \right)\nonumber \\
 & \geq & \left(\frac{1-2\delta_{s}}{1-\delta_{s}}\right),\label{eq:derivation_3}
\end{eqnarray}
}and (ii) for $j\in[M]\backslash\Omega_{i}$, 
\begin{eqnarray}
 &  & \left\Vert \bar{\mathbf{X}}(j)^{H}\left(\mathbf{I}-\mathbf{P}_{J}\right)\bar{\mathbf{X}}_{\Omega_{i}\backslash J}\right\Vert \nonumber \\
 & \leq & \left\Vert \bar{\mathbf{X}}(j)^{H}\bar{\mathbf{X}}_{\Omega_{i}\backslash J}\right\Vert +\nonumber \\
 &  & \left\Vert \bar{\mathbf{X}}(j)^{H}\bar{\mathbf{X}}_{J}\left((\bar{\mathbf{X}}{}_{J})^{H}\bar{\mathbf{X}}_{J}\right)^{-1}(\bar{\mathbf{X}}{}_{J})^{H}\bar{\mathbf{X}}_{\Omega_{i}\backslash J}\right\Vert \nonumber \\
 & \leq & \left(\delta_{s+1}+\frac{\delta_{s+1}\delta_{s}}{1-\delta_{s}}\right)\leq\left(\frac{\delta_{s+1}}{1-\delta_{s}}\right).\label{eq:derivation_4}
\end{eqnarray}
Recall from Definition \ref{Structured-Sparsity-Model} and (\ref{eq:transformed_H})
that $(\bar{\mathbf{H}}_{i})^{\Omega_{i}\backslash J}$ is i.i.d.
complex Gaussian distributed with zero mean and unit variance. From
(\ref{eq:derivation_3}), we obtain
\begin{equation}
\textrm{Pr}(\mathcal{D}_{i}^{[1]})\leq\Pr\left(\chi_{2N}\leq\frac{2N\alpha^{2}}{\left(\frac{1-2\delta_{s}}{1-\delta_{s}}\right)^{2}}\right).\label{eq:chi-square}
\end{equation}
\[
\textrm{Pr}(\mathcal{D}_{i}^{[3]})\leq\Pr\left(\chi_{2N}\leq\frac{2N}{\theta}\right)
\]
From (\ref{eq:derivation_4}), we obtain
\begin{equation}
\textrm{Pr}(\mathcal{D}_{i}^{[2]})\leq\sum_{j\notin\Omega_{i}}\cdot\Pr\left(\chi_{2N}\geq\frac{2N\left(\alpha-2\sqrt{\frac{\eta_{2}M(1+\delta_{1})}{P}}\right)^{2}}{\left(\frac{\delta_{s+1}}{1-\delta_{s}}\right)^{2}}\right).\label{eq:chi-square-2}
\end{equation}
For simplicity, we choose $\alpha\triangleq\sqrt{\left(\frac{\delta_{s+1}}{1-\delta_{s}}\right)+2\sqrt{\frac{\eta_{2}M(1+\delta_{1})}{P}}}\sqrt{\frac{1-2\delta_{s}}{1-\delta_{s}}}$.
From $\theta>1$, we obtain $\alpha\geq\frac{\delta_{s+1}}{1-\delta_{s}}+2\sqrt{\frac{\eta_{2}M(1+\delta_{1})}{P}}$.
Hence, from (\ref{eq:chi-square-2}), we further obtain 
\begin{equation}
\textrm{Pr}(\mathcal{D}_{i}^{[2]})\leq M\cdot\Pr\left(\chi_{2N}\geq\frac{2N\alpha^{2}}{\left(\frac{\delta_{s+1}}{1-\delta_{s}}+2\sqrt{\frac{\eta_{2}M(1+\delta_{1})}{P}}\right)^{2}}\right).\label{eq:relationship}
\end{equation}
From (\ref{eq:derivation_2}), (\ref{eq:chi-square}), (\ref{eq:chi-square-2})
and (\ref{eq:relationship}), 
\begin{eqnarray}
 &  & \Pr\left(I_{\{l\in\Omega_{i}^{'}\}}=0\right)\leq2\cdot\Pr\left(\chi_{2N}\leq\frac{2N}{\theta}\right)\label{eq:bound_orig}\\
 &  & \qquad+M\cdot\Pr\left(\chi_{2N}\geq2\theta N\right)+\textrm{Pr}\left(\chi_{2NT}\geq2\eta_{2}NT\right).\nonumber 
\end{eqnarray}

\par\end{flushleft}

Here, we would like to introduce the following inequalities from the
Chernoff bounds theory \cite{dasgupta2003elementary}:
\begin{lemma}
[Chernoff Bounds]\label{Chernoff-BoundsSuppose}Suppose $\chi_{2k}$
is chi-squared distributed with $2k$ degrees of freedom, we have
the following bound
\[
\Pr\left(\chi_{2k}\leq2xk\right)\leq\exp\left(-k\left(-1+x-\ln x\right)\right),\;0<x<1.
\]
\[
\Pr\left(\chi_{2k}\geq2xk\right)\leq\exp\left(-k\left(-1+x-\ln x\right)\right),\; x>1.
\]
\hfill \QED
\end{lemma}

Using the Chernoff bounds on (\ref{eq:bound_orig}), we obtain the
following result which completes the proof of equation (\ref{eq:bound_zero}):
{\small 
\begin{eqnarray*}
 &  & \Pr\left(I_{\{l\in\Omega_{i}^{'}(J)\}}=0\right)\leq2\exp\left(-N\left(\ln\theta+\frac{1}{\theta}-1\right)\right)+M\times\\
 &  & \exp\left(-N\left(\theta-\ln\theta-1\right)\right)+\exp\left(-NT\left(\eta_{2}-\ln\eta_{2}-1\right)\right)<p.
\end{eqnarray*}
}{\small \par}

Next, we prove (\ref{eq:bound_zero_II}). First, we have {\small 
\begin{eqnarray}
\Pr\left(\max_{j\in[M]\backslash\Omega_{i}}I_{\{j\in\Omega_{i}^{'}(J)\}}=1\right)\overset{(b_{2})}{\leq}\textrm{Pr}\left(\vphantom{\min_{l\in\Omega_{i}\backslash J}\left\Vert \bar{\mathbf{X}}(l)^{H}\left(\mathbf{I}-\mathbf{P}_{J}\right)\mathbf{\bar{Y}}_{i}\right\Vert }\qquad\qquad\qquad\qquad\right.\label{eq:further1}\\
\left.\max_{j\in[M]\backslash\Omega_{i}}\left\Vert \bar{\mathbf{X}}(j)^{H}\left(\mathbf{I}-\mathbf{P}_{J}\right)\mathbf{\bar{Y}}_{i}\right\Vert \geq\sqrt{\eta_{1}N}\right),\nonumber 
\end{eqnarray}
}where $(b_{2})$ comes from Lemma \ref{Event-Implications} and the
fact that {\small 
\[
\max_{j\in[M]\backslash\Omega_{i}}\left\Vert \bar{\mathbf{X}}(j)^{H}\left(\mathbf{I}-\mathbf{P}_{J}\right)\mathbf{\bar{Y}}_{i}\right\Vert _{F}<\sqrt{\eta_{1}N}
\]
}would imply $\max_{j\in[M]\backslash\Omega_{i}}I_{\{j\in\Omega_{i}^{'}\}}=0$
according to (\ref{eq:omega_pi}). Based on (\ref{eq:further1}),
and similar to (\ref{eq:derivation_2}), we obtain
\begin{eqnarray}
 &  & \Pr\left(\max_{j\in[M]\backslash\Omega_{i}}I_{\{j\in\Omega_{i}^{'}(J)\}}=1\right)\leq\label{eq:property}\\
 &  & \Pr\left(\max_{j\in[M]\backslash\Omega_{i}}\left\Vert \bar{\mathbf{X}}(j)^{H}\left(\mathbf{I}-\mathbf{P}_{J}\right)\bar{\mathbf{X}}_{\Omega_{i}\backslash J}\left(\bar{\mathbf{H}}_{i}\right)^{\Omega_{i}\backslash J}\right\Vert \geq\right.\nonumber \\
 &  & \left.\sqrt{\eta_{1}N}-\sqrt{\frac{\eta_{2}MN(1+\delta_{1})}{P}}\right)+\textrm{Pr}\left(\chi_{2NT}\geq2\eta_{2}NT\right)\nonumber 
\end{eqnarray}
{\small Based on (\ref{eq:property})}, and similar to (\ref{eq:chi-square-2}),
we obtain
\begin{eqnarray*}
 &  & \Pr\left(\max_{j\in[M]\backslash\Omega_{i}}I_{\{j\in\Omega_{i}^{'}\}}=1\right)\leq\exp\left(-NT\left(\eta_{2}-\ln\eta_{2}-1\right)\right)\\
 &  & +M\cdot\exp\left(-N\left(\theta-\ln\theta-1\right)\right)\leq p.
\end{eqnarray*}
Therefore, (\ref{eq:bound_zero_II}) in Lemma \ref{Probability-Bounds-I}
is proved.

\subsection{\label{sub:Proof-of-Theorem-individual}Proof of Theorem \ref{Probability-Bounds-of-individual}}

Given event $\Theta_{c}$ and $\Lambda$, we then investigate the
probability that the estimated support for user $i$ is correct, i.e.,
$\Omega_{i}^{e}=\Omega_{i}$ (event $\Theta_{i}$). First, from (\ref{eq:random_sele}),
any selected index will not be selected again by Step 3. A. Hence,
$\Theta_{i}\mid\Theta_{c}\Lambda$ happens if and only if for user
$i$, Step 3 adds $|\Omega_{i}|-s_{c}$ \emph{new} indices belonging
to $\Omega_{i}$ and then \emph{stops}. We first find a sufficient
condition for event $\Theta_{i}$ to happen as follows:
\begin{lemma}
[Sufficient Conditions for $\Theta_{i}\mid\Theta_{c}\Lambda$]\label{Sufficient-Condition-for-theta_i}If
conditions $\mathcal{N}$, $\mathcal{S}$ holds and event $\mathcal{E}_{J}$
happens for all $J\subseteq\Omega_{i}$, $s_{c}\leq|J|<|\Omega_{i}|$,
then event $\Theta_{i}\mid\Theta_{c}\Lambda$ will surely happen,
where
\[
\mathcal{N}:\quad\left\Vert \mathbf{\bar{N}}_{i}\right\Vert _{F}^{2}\leq\frac{\eta_{2}MN}{P},
\]
\begin{equation}
\mathcal{S}:\quad\min_{l\in\Omega_{i}}||\left(\bar{\mathbf{H}}_{i}\right)^{\{l\}}||_{F}>\frac{2}{\sqrt{1-\delta_{s}}}\sqrt{\frac{\eta_{2}NM}{P}},\label{eq:equivalentcon_s}
\end{equation}
\begin{eqnarray}
\mathcal{E}_{J}:\quad &  & \max_{l\in\Omega_{i}\backslash J}\left\Vert \mathbf{\bar{X}}(l)^{H}\left(\mathbf{I}-\mathbf{P}_{J}\right)\mathbf{\bar{Y}}_{i}\right\Vert \label{eq:first}\\
 &  & >\max_{j\in[M]\backslash\Omega_{i}}\left\Vert \mathbf{\bar{X}}(j)^{H}\left(\mathbf{I}-\mathbf{P}_{J}\right)\mathbf{\bar{Y}}_{i}\right\Vert .\nonumber 
\end{eqnarray}
\end{lemma}
\begin{proof}
We prove the lemma by following the procedures in Step 3. At first,
the estimated support is $\Omega_{i}^{e}=\Omega_{c}^{e}\subseteq\Omega_{i}$
and $|\Omega_{i}^{e}|<|\Omega_{i}|$; then Step 3 will add a new support
index instead of stopping for user $i$ because 
\begin{eqnarray*}
 &  & \left\Vert \left(\mathbf{I}-\mathbf{P}_{\Omega_{i}^{e}}\right)\mathbf{\bar{Y}}_{i}\right\Vert _{F}\\
 & \overset{(c_{1})}{\geq} & \left\Vert \left(\mathbf{I}-\bar{\mathbf{X}}_{\Omega_{i}^{e}}(\bar{\mathbf{X}}_{\Omega_{i}^{e}})^{\dagger}\right)\bar{\mathbf{X}}_{\Omega_{i}}\bar{\mathbf{H}}_{i}^{\Omega_{i}}\right\Vert _{F}-\sqrt{\frac{\eta_{2}NM}{P}}\\
 & \overset{(c_{2})}{=} & \left\Vert \bar{\mathbf{X}}_{\Omega_{i}}\left(\bar{\mathbf{H}}_{i}^{a}-\bar{\mathbf{H}}_{i}\right)^{\Omega_{i}}\right\Vert _{F}-\sqrt{\frac{\eta_{2}NM}{P}}\\
 & \overset{(c_{3})}{\geq} & \sqrt{1-\delta_{s}}\min_{l\in\Omega_{i}}||\left(\bar{\mathbf{H}}_{i}\right)^{\{l\}}||_{F}-\sqrt{\frac{\eta_{2}NM}{P}}\\
 & > & \sqrt{\frac{\eta_{2}NM}{P}},
\end{eqnarray*}
where $(c_{1})$ comes from condition $\mathcal{N}$; ($c_{2}$) comes
from $\Omega_{i}^{e}\subseteq\Omega_{i}$ and $(\bar{\mathbf{H}}_{i}^{a})^{\Omega_{i}^{e}}=(\bar{\mathbf{X}}_{\Omega_{i}^{e}})^{\dagger}\bar{\mathbf{X}}_{\Omega_{i}}\bar{\mathbf{H}}_{i}^{\Omega_{i}}$,
$(\bar{\mathbf{H}}_{i}^{a})^{[M]\backslash\Omega_{i}^{e}}=\mathbf{0}$;
and $(c_{3})$ comes from $|\Omega_{i}^{e}|<|\Omega_{i}|$. On the
other hand, from Step 3. A, the added index will be belonging to $\Omega_{i}\backslash\Omega_{i}^{e}$
as $\mathcal{E}_{J}$ holds for all $J$. Following the above procedures,
Step 3 will continuously add indices that belong to $\Omega_{i}\backslash\Omega_{i}^{e}$
until $|\Omega_{i}^{e}|=|\Omega_{i}|$. Suppose Step 3 run into the
case of $\Omega_{i}^{e}=\Omega_{i}$; then,
\[
||\mathbf{R}_{i}||_{F}^{2}=\left\Vert \left(\mathbf{I}-\mathbf{P}_{\Omega_{i}}\right)\mathbf{\bar{N}}_{i}\right\Vert _{F}^{2}\leq\left\Vert \mathbf{\bar{N}}_{i}\right\Vert _{F}^{2}\leq\frac{\eta_{2}MN}{P}
\]
and hence Step 3 stops for user $i$. Therefore, under the conditions
in Lemma \ref{Sufficient-Condition-for-theta_i}, event $\Theta_{i}$
conditioned on $\Theta_{c}$ and $\Lambda$ will surely happen. 
\end{proof}

From Lemma \ref{Sufficient-Condition-for-theta_i} and Lemma \ref{Event-Implications},
we obtain:{\small 
\begin{eqnarray}
 &  & \textrm{Pr}(\bar{\Theta}_{i})\leq\textrm{Pr}\left(\min_{l\in\Omega_{i}}||\left(\bar{\mathbf{H}}_{i}\right)^{\{l\}}||_{F}\leq\frac{2\sqrt{\frac{\eta_{2}NM}{P}}}{\sqrt{1-\delta_{s}}}\right)\label{eq:hoby}\\
 &  & +\sum_{J\subseteq\Omega_{i},s_{c}\leq|J|<|\Omega_{i}|}\textrm{Pr}\left(\mathcal{\bar{E}}_{J}\right)+\textrm{Pr}\left(\left\Vert \mathbf{\bar{N}}_{i}\right\Vert _{F}^{2}>\frac{\eta_{2}MN}{P}\right).\nonumber 
\end{eqnarray}
}First, we obtain
\begin{eqnarray}
 &  & \textrm{Pr}\left(\min_{l\in\Omega_{i}}||\left(\bar{\mathbf{H}}_{i}\right)^{\{l\}}||_{F}\leq\frac{2\sqrt{\frac{\eta_{2}NM}{P}}}{\sqrt{1-\delta_{s}}}\right)\nonumber \\
 & \leq & \sum_{l\in\Omega_{i}}\textrm{Pr}\left(||\left(\bar{\mathbf{H}}_{i}\right)^{\{l\}}||_{F}^{2}\leq\frac{4\eta_{2}NM}{(1-\delta_{s})P}\right)\nonumber \\
 & \leq & s\cdot\textrm{Pr}\left(\chi_{2N}\leq\frac{8\eta_{2}NM}{(1-\delta_{s})P}\right)\nonumber \\
 & \overset{(c_{4})}{\leq} & s\exp\left(-N\left(\ln\vartheta+\frac{1}{\vartheta}-1\right)\right),\label{eq:ressult1}
\end{eqnarray}
where $(c_{4})$ uses the Chernoff bounds in Lemma \ref{Chernoff-BoundsSuppose}.
Second, from Lemma \ref{Suppose--and}, for a given $J\subseteq\Omega_{i},s_{c}\leq|J|<|\Omega_{i}|$:
{\small 
\begin{eqnarray}
 &  & \textrm{Pr}\left(\mathcal{\bar{E}}_{J}\right)\leq\nonumber \\
 &  & \Pr\left(\max_{l\in\Omega_{i}\backslash J}\left\Vert \mathbf{\bar{X}}(l)^{H}\left(\mathbf{I}-\mathbf{P}_{J}\right)\mathbf{\bar{X}}_{\Omega_{i}\backslash J}\left(\bar{\mathbf{H}}_{i}\right)^{\Omega_{i}\backslash J}\right\Vert \leq\sqrt{N}\alpha\right)+\nonumber \\
 &  & \Pr\left(\max_{j\in[M]\backslash\Omega_{i}}\left\Vert \mathbf{\bar{X}}(j)^{H}\left(\mathbf{I}-\mathbf{P}_{J}\right)\mathbf{\bar{X}}_{\Omega_{i}\backslash J}\left(\bar{\mathbf{H}}_{i}\right)^{\Omega_{i}\backslash J}\right\Vert \right.\nonumber \\
 &  & \left.\geq\sqrt{N}\alpha-2\sqrt{\frac{\eta_{2}MN(1+\delta_{1})}{P}}\right).\label{eq:similar_0}
\end{eqnarray}
}From (\ref{eq:derivation_3}) and (\ref{eq:derivation_4}), and similar
to the derivation in Appendix \ref{sub:Proof-of-Lemma-bound_I}, we
obtain {\small 
\begin{eqnarray}
 &  & \textrm{Pr}\left(\mathcal{\bar{E}}_{J}\right)\leq\exp\left(-N\left(\ln\theta+\frac{1}{\theta}-1\right)\right)\label{eq:similar_2}\\
 &  & +M\cdot\exp\left(-N\left(\theta-\ln\theta-1\right)\right).\nonumber 
\end{eqnarray}
}Substituting (\ref{eq:ressult1}) and (\ref{eq:similar_2}) into
(\ref{eq:hoby}), we obtain the desired theorem.

\subsection{\label{sub:Proof-of-Theorem-distortion}Proof of Theorem \ref{CSIT-Distortion-under}}

First note that $\frac{\left\Vert \mathbf{H}_{i}-\mathbf{H}_{i}^{e}\right\Vert _{F}}{\left\Vert \mathbf{H}_{i}\right\Vert _{F}}=\frac{\left\Vert \mathbf{\bar{H}}_{i}-\mathbf{\bar{H}}^{e}\right\Vert _{F}}{\left\Vert \mathbf{\bar{H}}_{i}\right\Vert _{F}}$
and {\small 
\begin{eqnarray}
 &  & \mathbb{E}\left(\frac{\left\Vert \mathbf{\bar{H}}_{i}-\mathbf{\bar{H}}^{e}\right\Vert _{F}}{\left\Vert \mathbf{\bar{H}}_{i}\right\Vert _{F}}\right)=\mathbb{E}\left(\frac{\left\Vert \mathbf{\bar{H}}_{i}-\mathbf{\bar{H}}^{e}\right\Vert _{F}}{\left\Vert \mathbf{\bar{H}}_{i}\right\Vert _{F}}\mid\bar{\Lambda}\right)\Pr(\bar{\Lambda})\nonumber \\
 &  & +\mathbb{E}\left(\frac{\left\Vert \mathbf{\bar{H}}_{i}-\mathbf{\bar{H}}^{e}\right\Vert _{F}}{\left\Vert \mathbf{\bar{H}}_{i}\right\Vert _{F}}\mid\Theta_{i}\Theta_{c}\Lambda\right)\Pr(\Theta_{i}\Theta_{c}\Lambda)+\nonumber \\
 &  & \mathbb{E}\left(\frac{\left\Vert \mathbf{\bar{H}}_{i}-\mathbf{\bar{H}}^{e}\right\Vert _{F}}{\left\Vert \mathbf{\bar{H}}_{i}\right\Vert _{F}}\mid\overline{\Theta_{i}\Theta_{c}}\Lambda\right)\Pr(\overline{\Theta_{i}\Theta_{c}}\mid\Lambda)\Pr\left(\Lambda\right).\label{eq:Interme}
\end{eqnarray}
}Second, {\small 
\begin{eqnarray}
 &  & \frac{\left\Vert \mathbf{\bar{H}}_{i}-\mathbf{\bar{H}}^{e}\right\Vert _{F}}{\left\Vert \mathbf{\bar{H}}_{i}\right\Vert _{F}}\leq\frac{\left\Vert \left(\bar{\mathbf{X}}_{\Omega_{i}^{e}}\right)^{\dagger}\mathbf{\bar{N}}_{i}\right\Vert _{F}}{\left\Vert \mathbf{\bar{H}}_{i}\right\Vert _{F}}+\frac{\left\Vert \mathbf{\bar{H}}_{i}^{\Omega_{i}\backslash\Omega_{i}^{e}}\right\Vert _{F}}{\left\Vert \mathbf{\bar{H}}_{i}\right\Vert _{F}}\nonumber \\
 &  & +\frac{1}{\left\Vert \mathbf{\bar{H}}_{i}\right\Vert _{F}}\left\Vert \left(\bar{\mathbf{X}}_{\Omega_{i}^{e}}\right)^{\dagger}\bar{\mathbf{X}}_{\Omega_{i}\backslash\Omega_{i}^{e}}(\mathbf{\bar{H}}_{i})^{\Omega_{i}\backslash\Omega_{i}^{e}}\right\Vert _{F}\label{eq:result2}
\end{eqnarray}
}Note that $\left(\bar{\mathbf{X}}_{\Omega_{i}^{e}}\right)^{\dagger}$
contains at most $s$ non-zero singular values (i.e., $|\Omega_{i}^{e}|\leq s$)
and each non-zero singular value is upper bounded by $\frac{1}{\sqrt{1-\delta_{s}}}$
from (\ref{eq:property6}). Furthermore, $\mathbf{\bar{N}}_{i}$ and
$(\mathbf{\bar{H}}_{i})^{\Omega_{i}}$ are i.i.d. complex Gaussian
distributed (with variance $\frac{M}{PT}$ and $1$ respectively).
Hence, {\small 
\begin{eqnarray*}
 &  & \mathbb{E}\left(\frac{\left\Vert \left(\bar{\mathbf{X}}_{\Omega_{i}^{e}}\right)^{\dagger}\mathbf{\bar{N}}_{i}\right\Vert _{F}}{\left\Vert (\mathbf{\bar{H}}_{i})^{\Omega_{i}}\right\Vert _{F}}\right)\leq\sqrt{\frac{2MsN}{PT(1-\delta_{s})}}\mathbb{E}\left((\chi_{2\tilde{s}_{i}N})^{-\frac{1}{2}}\right)\leq\\
 &  & \sqrt{\frac{2MsN}{PT(1-\delta_{s})}}\mathbb{E}\left((\chi_{2N})^{-\frac{1}{2}}\right)\leq\sqrt{\frac{MsN}{PT(1-\delta_{s})}}\frac{\Gamma\left(N-\frac{1}{2}\right)}{\Gamma\left(N\right)},
\end{eqnarray*}
}where $\chi_{2\tilde{s}_{i}N}$ denotes the chi-distribution with
$2\tilde{s}_{i}N$ degrees of freedom, where $\tilde{s}_{i}\triangleq|\Omega_{i}|\geq1$
from Definition \ref{Structured-Sparsity-Model}. When $\Theta_{i}\Theta_{c}\Lambda$
happens, $\Omega_{i}=\Omega_{i}^{e}$, and hence $\mathbf{\bar{H}}_{i}^{\Omega_{i}\backslash\Omega_{i}^{e}}=\mathbf{0}$
in (\ref{eq:result2}). Note that $\left\Vert \mathbf{\bar{H}}_{i}^{\Omega_{i}\backslash\Omega_{i}^{e}}\right\Vert _{F}\leq\left\Vert \mathbf{\bar{H}}_{i}\right\Vert _{F}$
always holds. When $\overline{\Theta_{i}\Theta_{c}}\mid\Lambda$ happens,
from $\left\Vert \left(\bar{\mathbf{X}}_{\Omega_{i}^{e}}\right)^{\dagger}\bar{\mathbf{X}}_{\Omega_{i}\backslash\Omega_{i}^{e}}\right\Vert \leq\frac{\delta_{2s}}{1-\delta_{s}}$
according to Lemma \ref{For-given-inequality}, we obtain 
\begin{equation}
\frac{1}{\left\Vert \mathbf{\bar{H}}_{i}\right\Vert _{F}}\left\Vert \left(\bar{\mathbf{X}}_{\Omega_{i}^{e}}\right)^{\dagger}\bar{\mathbf{X}}_{\Omega_{i}\backslash\Omega_{i}^{e}}(\mathbf{\bar{H}}_{i})^{\Omega_{i}\backslash\Omega_{i}^{e}}\right\Vert \leq\frac{\delta_{2s}}{1-\delta_{s}}\label{eq:final_2}
\end{equation}
When $\bar{\Lambda}$ happens, from $\left\Vert \left(\bar{\mathbf{X}}_{\Omega_{i}^{e}}\right)^{\dagger}\right\Vert \leq\frac{1}{\sqrt{1-\delta_{s}}}$,
$||\bar{\mathbf{X}}(j)||_{F}\leq\sqrt{1+\delta_{1}}$, $\forall j$,
and the fact that $(\mathbf{\bar{H}}_{i})^{\Omega_{i}\backslash\Omega_{i}^{e}}$
are i.i.d. complex Gaussian distributed, we obtain {\small 
\begin{eqnarray}
 &  & \mathbb{E}\left(\frac{\left\Vert \left(\bar{\mathbf{X}}_{\Omega_{i}^{e}}\right)^{\dagger}\bar{\mathbf{X}}_{\Omega_{i}\backslash\Omega_{i}^{e}}(\mathbf{\bar{H}}_{i})^{\Omega_{i}\backslash\Omega_{i}^{e}}\right\Vert _{F}}{\left\Vert \mathbf{\bar{H}}_{i}\right\Vert _{F}}\right)\label{eq:final3}\\
 & \leq & \frac{1}{\sqrt{1-\delta_{s}}}\mathbb{E}\left(\frac{\left\Vert \bar{\mathbf{X}}_{\Omega_{i}\backslash\Omega_{i}^{e}}(\mathbf{\bar{H}}_{i})^{\Omega_{i}\backslash\Omega_{i}^{e}}\right\Vert _{F}}{\left\Vert (\mathbf{\bar{H}}_{i})^{\Omega_{i}\backslash\Omega_{i}^{e}}\right\Vert _{F}}\right)\leq\sqrt{\frac{1+\delta_{1}}{1-\delta_{s}}}\nonumber 
\end{eqnarray}
}Furthermore, we have $\Pr(\overline{\Theta_{i}\Theta_{c}}\mid\Lambda)\leq2-\textrm{Pr}\left(\Theta_{i}\mid\Theta_{c}\Lambda\right)-\textrm{Pr}\left(\Theta_{c}\mid\Lambda\right)$.
Substituting (\ref{eq:result2}), (\ref{eq:final_2}) and (\ref{eq:final3})
into (\ref{eq:Interme}), we obtain the desired theorem.

\subsection{\label{sub:Proof-of-Theorem-scalen}Proof of Corollary \ref{Asymptotic-Probability-Bounds-I-1}}

From (\ref{eq:p_omegac}), we obtain $C_{i}\leq g_{0}p^{\left\lfloor \frac{1-\gamma}{2}K\right\rfloor }$,
where $g_{0}=2c_{0}\left(\frac{1-\delta_{s}+\delta_{2s}}{1-\delta_{s}}\right)\sum_{t=0}^{\left\lceil \frac{1+\gamma}{2}K\right\rceil }\left(\begin{array}{c}
K\\
t
\end{array}\right)$ which does not depend on $N$. Hence, from (\ref{eq:ppp}), {\small 
\[
\lim_{N\rightarrow\infty}-\frac{1}{N}\ln\left(C_{i}\right)\geq\lim_{N\rightarrow\infty}-\left\lfloor \frac{1-\gamma}{2}K\right\rfloor \frac{\log p}{N}\geq\left\lfloor \frac{1-\gamma}{2}K\right\rfloor \beta_{1}.
\]
}On the other hand, (\ref{eq:R_d_definition-1}) can be obtained from
(\ref{eq:p_omegai}) similarly.

\subsection{\label{sub:Proof-of-Corollary-2}Proof of Corollary \ref{Asymptotic-Probability-Bounds-I}}

From the large deviation result on Bernoulli random variables \cite{varadhan1988large},
we obtain the following Lemma:
\begin{lemma}
\label{Suppose-that-large}Suppose $0\leq p<1-\frac{K_{2}}{K}<1$,
then 
\begin{eqnarray}
 &  & \lim_{K\rightarrow\infty}-\frac{1}{K}\ln\sum_{t=0}^{K_{2}}\left(\begin{array}{c}
K\\
t
\end{array}\right)(1-p)^{t}p^{K-t}\nonumber \\
 &  & \quad=\left(1-\frac{K_{2}}{K}\right)\ln\frac{p(K-K_{2})}{K_{2}(1-p)}-\ln\frac{Kp}{K_{2}}>0.\label{eq:result}
\end{eqnarray}
\end{lemma}
\begin{proof}
Suppose that $\{z_{i}\}_{i=1}^{M}$ is a series of i.i.d. Bernoulli
random variables with $\Pr(z_{i}=1)=p$ and $\Pr(z_{i}=0)=1-p$, $\forall i$.
We have
\begin{eqnarray}
 &  & \lim_{K\rightarrow\infty}-\frac{1}{K}\ln\sum_{t=0}^{K_{2}}\left(\begin{array}{c}
K\\
t
\end{array}\right)(1-p)^{t}p^{K-t}\nonumber \\
 & = & \lim_{K\rightarrow\infty}-\frac{1}{K}\ln\Pr\left(\sum_{i=1}^{K}z_{i}\geq K-K_{2}\right)\nonumber \\
 & \overset{(e_{1})}{=} & \sup_{\epsilon>0}\left(\frac{(K-K_{2})\epsilon}{K}-\mathbb{E}\ln\left(\exp\left(\epsilon z_{i}\right)\right)\right)\label{eq:large_deviation}
\end{eqnarray}
where $(e_{1})$ is from the large deviation theory \cite{varadhan1988large}
and (\ref{eq:large_deviation}) can be easily simplified to be (\ref{eq:result}). 
\end{proof}

Substituting the lower bound of $\textrm{Pr}(\Theta_{c}\mid\Lambda)$
in Theorem \ref{Correct-Support-Recovery-I} into (\ref{eq:Distortion})
and further using Lemma \ref{Suppose-that-large} with $K_{2}=\left\lceil \frac{1+\gamma}{2}K\right\rceil $,
we obtain Corollary \ref{Asymptotic-Probability-Bounds-I}. 

\bibliographystyle{IEEEtran}
\bibliography{CSIT_REF}

\end{document}